\documentclass[reqno]{amsart}

\usepackage{amssymb,amsmath,amsthm,longtable,verbatim}

\usepackage{graphicx}

\def\lt{\biggl}                  \def\rt{\biggr}

\newcommand{\C}{{\mathbb C}}

\newcommand{\Z}{{\mathbb Z}}\newcommand{\E}{{\mathbb E}}

\newcommand{\1}{{ \mathbf  1}}

\newcommand{\R}{{\mathbb R}}

\newcommand{\const}{\mathrm{const}\, }

\renewcommand{\Pr}{\mathbf{Pr}}

\newcommand{\mesh}{\delta}
\newcommand{\Od}{\Omega^\mesh}
\newcommand{\dOd}{\widetilde{\Omega}^\mesh}
\newcommand{\dOm}{\widetilde{\Omega}}

\newcommand{\Cvr}{\varpi}
\newcommand{\Corn}{\Upsilon}

\newcommand\Conf{\mathrm{Conf}}

\renewcommand\mod{\mathrm{mod}}
\newcommand\wind{\mathrm{w}}
\newcommand\pa{\partial}
\newcommand\wt{\widetilde}
\newcommand\vare{\varepsilon}
\renewcommand\Re{\mathrm{Re}}
\renewcommand\Im{\mathrm{Im}}

\def\cF{{\mathcal F}}
\def\cH{{\mathcal H}}
\def\rZ{{\mathrm Z}}

\def\res{\mathop{\mathrm{res}}\limits}

\def\cB{{\mathcal B}}

\newcommand\dist{\mathrm{dist}}

\theoremstyle{remark}

\theoremstyle{plain}
\newtheorem{theorem}{Theorem}[section]
\newtheorem*{theorem*}{Theorem}

\newtheorem*{corollary*}{Corollary}
\newtheorem{definition}[theorem]{Definition}

\newtheorem{lemma}[theorem]{Lemma} %[section]
\newtheorem{corollary}[theorem]{Corollary} %[section]
\newtheorem{proposition}[theorem]{Proposition} %[section]

\theoremstyle{remark}
\newtheorem{rem}[theorem]{Remark}

\newcommand\reflist[1]{#1$\vphantom{a}^\circ$}

\date{}
\title[Holomorphic spinor observables in the Ising model]
{Holomorphic spinor observables\\ in the critical Ising model}

\author[D. Chelkak]{Dmitry Chelkak$^\mathrm{a,c}$}
\thanks{}
\email{}

\author[K. Izyurov]{Konstantin Izyurov$^\mathrm{b,c}$}

\thanks{\textsc{${}^\mathrm{A}$ St.Petersburg Department of Steklov Mathematical Institute (PDMI RAS).
Fontanka~27, 191023 St.Petersburg, Russia.}}

\thanks{\textsc{${}^\mathrm{B}$ Section de Math\'ematiques, Universit\'e de Gen\`eve.
2-4 rue du Li\`evre, Case postale~64, 1211 Gen\`eve 4, Suisse. {\it Current address:}
Department of Mathematics and Statistics, University of Helsinki, P.O. box 68  00014 Helsinki, Finland }}

\thanks{\textsc{${}^\mathrm{C}$ Chebyshev Laboratory, Department of Mathematics and
Mechanics, Saint-Petersburg State University. 14th Line, 29b, 199178 Saint-Petersburg,
Russia.}}

\thanks{{\it E-mail addresses:} \texttt{dchelkak@pdmi.ras.ru, konstantin.izyurov@helsinki.fi}}

\begin{document}

\begin{abstract}
We introduce a new version of discrete holomorphic observables for the critical planar Ising
model. These observables are holomorphic spinors defined on double covers of the original
multiply connected domain. We compute their scaling limits, and show their relation to the
ratios of spin correlations, thus providing a rigorous proof to a number of formulae for those ratios predicted by CFT arguments.
\end{abstract}

\maketitle

\newcounter{Listcounter}

\section{Introduction}

\subsection{The critical Ising model and Smirnov's holomorphic observables.}
The two-dimensional Ising model is one of the most well-studied models in statistical mechanics. Given a discrete planar domain $\Od$ (a bounded subset of the square grid), the Ising model in $\Od$ can be viewed either as a random assignment of spins to the \emph{faces} of $\Od$, or a random collection of edges of $\Od$, with an edge drawn between each pair of faces having different spins. The partition function of the model is given by
\[
\sum\limits_{\sigma:\mathcal{F}(\Od)\to\{-1;1\}} \exp\lt[\frac{1}{T}\sum_{f\sim f'}\sigma(f)\sigma(f')\rt]
\quad
\text{ or }
\quad
\sum\limits_{S\in\Conf(\Od)} x^{|S|},
\]
respectively, where $\mathcal{F}(\Od)$ denotes the set of faces of $\Od$ and $\Conf(\Od)$ is the set of subgraphs $S$ of $\Od$ such that all vertices of $\Od$ have even degrees in $S$. We refer the reader to Section~\ref{Sec: notations} for a more detailed discussion and notation. We will be interested in the properties of the model at the critical temperature $T=2\log^{-1}(\sqrt{2}+1)$, which corresponds to $x=\sqrt{2}-1$. This value of $x$ will be fixed throughout the paper.

Discrete holomorphic observables, also called holomorphic fermions or fermionic observables, were proposed by Smirnov in~\cite{Smirnov06} as a tool to study the critical Ising model, although similar objects appeared earlier in~\cite{KadanoffCeva} and~\cite{Mercat01} without discussing
corresponding boundary value problems. Since then, these observables proved to be very useful for a
rigorous analysis of the planar Ising model at criticality in the scaling limit
when $\Od$ approximates some continuous domain $\Omega$ as the lattice mesh $\mesh$ tends to
zero.

Recall that Smirnov's fermionic observable is defined as %was defined in \cite{Smirnov06} as
\begin{equation}
\label{SmirnovF}
F(a,z):= (-in_a)^{-1/2}\!\! \sum_{S\in\Conf_{a,z}}x^{|S|}e^{-i\wind(\gamma)/2},
 \end{equation}
where $\Conf_{a,z}$ is the set of edge subsets $S$, such that $S$ can be decomposed into a disjoint collection of loops and a simple lattice path $\gamma$ connecting a boundary edge $a$ to the midpoint $z=z_e$ of an inner edge $e$; $\wind(\gamma)$ is the winding number of~$\gamma$; and %$\eta_a=(-in_a)^{1/2}\in\{\pm 1, \pm e^{i\frac{\pi}{4}},\pm i, \pm e^{i\frac{3\pi}{4}}\}$
$n_a\in\{\pm 1,\pm i\}$ denotes the orientation of the outgoing boundary edge $a$. With this definition, the observable has been shown to be discrete holomorphic and satisfy Riemann-type boundary conditions
\begin{equation}
\label{eq: rim_bc}
F(a,z)\sqrt{in_z}\in \R, \quad z\in\partial\Od\backslash \{a\}.
\end{equation}
This led to a proof of its convergence to a conformally covariant scaling limit \cite{CHS2}.

This result has been the main ingredient of the recent progress in rigorous understanding of conformal invariance in the critical two-dimensional Ising model. The martingale property of $F(a,z)$ (see further details in~\cite{CHS2}) allows one to prove convergence of the Ising interfaces to the chordal
Schramm's $\mathrm{SLE}_3$ curves. Using a slightly different version of this observable, Hongler and Smirnov~\cite{HonSmi} were
able to compute the scaling limit of the energy density in the critical Ising model on the
square grid, including the lattice dependent constant before the conformally covariant
factor. This result was later extended to all correlations of the energy density field and
certain boundary spin correlations \cite{HThesis}.

At the same time, similar observables
proved to be very useful in the analysis of the random cluster (Fortuin-Kasteleyn)
representation of the critical Ising model ~\cite{Smirnov06, RivaCardy06, Smirnov10,CHS2,DHNolin}.
In particular, it was shown by Beffara and Duminil-Copin~\cite{BeffaraDuminil10} that they
can be used to give a short proof of criticality of the Ising model at the self-dual point.

Many of these results generalize beyond the case of square grid approximations. Thus, convergence of fermionic observables has been proven for isoradial lattices \cite{CHS2}, which reappeared in the connection with the critical Ising model in the paper of
Mercat~\cite{Mercat01}. This proved the universality phenomenon, i.e., the fact that a
microscopic structure of the lattice does not affect macroscopic properties of the scaling
limit. Moreover, discrete complex analysis technique developed in \cite{CHS} and~\cite{CHS2} provides a general framework for such universal proofs.

On the other hand, one of the most natural questions about the Ising model -- the rigorous proof of conformal covariance of spin correlations in the scaling limit -- remained out of reach until recently. The goal of the present work is to introduce a new tool -- spinor holomorphic observables -- that allows to attack this problem. In particular, we prove convergence of ratios of spin correlations corresponding to different boundary conditions to conformally invariant limits. In a subsequent joint paper with Cl\'ement Hongler \cite{CHHI}, using a more elaborate version of the spinor observables, we prove conformal covariance of spin correlations themselves.

\subsection{Spinor holomorphic observables and ratios of spin correlations}

In this paper we extend the study of fermionic observables to the case  of multiply connected
domains. Given a double cover $\Cvr:\dOd\to\Od$ of such a domain, we define the observable $F_\Cvr(a,\cdot):\dOd\to\C$ by
\begin{equation}
F_\Cvr(a,z):=(-in_a)^{-1/2}\!\! \sum_{S\in\Conf_{\Cvr(a),\Cvr(z)}}x^{|S|}e^{-i\wind(\gamma)/2}(-1)^{l(S)+\1_{\gamma:a\to z}},
\end{equation}
where $a,z\in \dOd$, but the sum is taken over \emph{the same set of configurations} as before; $l(s)$ is the number of loops in $S$ that do not lift as closed loops to $\dOd$, and $\1_{\gamma:a\to z}$ is the indicator of the event that $\gamma$ lifts to $\dOd$ as a path running from $a$ to $z$ (and not to the other sheet), see Section~\ref{Sec: observables_limits} for detailed discussion. In other words, we plug into (\ref{SmirnovF}) an additional sign that depends on homology class of $S$ modulo two. It is worth to mention that our observables should be closely related to the vector bundle Laplacian technique applied to uniform spanning trees and double dimers by Kenyon~\cite{Ken10, Ken11}, although at the moment we do not know any exact correspondence of that sort.

Our main observation is that $F_\Cvr(a,z)$ are discrete holomorphic and satisfy the boundary conditions (\ref{eq: rim_bc}), just like Smirnov's observable $F(a,z)$. The definition implies that $F_\Cvr(a,z)=-F_\Cvr(a,z^{\ast})$, if $z\ne z^{\ast}$ belong to a fiber of the same point; hence, we call $F_\Cvr$ \emph{holomorphic spinors}.

To describe the scaling limits of $F_\Cvr(a,\cdot)$, we will introduce the continuous holomorphic spinors $f_\Cvr(a,\cdot)$. Roughly speaking, these are fundamental solutions to the continuous Riemann boundary value problem (\ref{eq: rim_bc}) on the double-cover $\wt{\Omega}$, with a singularity at $a$ and the property $f_\Cvr(z)\equiv -f_\Cvr(z^{\ast})$.  Postponing precise definitions until Section~\ref{Sec: observables_limits}, we now state our first main result (see Theorem \ref{thm: mainconv_int_1}):

{\renewcommand{\thetheorem}{\Alph{theorem}}

\begin{theorem}
Suppose that $\Od$ is a sequence of discrete domains of mesh size $\delta$ approximating (in the sense of Carath\'eodory) a continuous finitely connected domain $\Omega$, and that $a^\mesh\in \pa\Od$ tends to $a\in\pa\Omega$ as $\mesh\to 0$. Then there is a sequence of normalizing factors $\beta(\delta)=\beta(\delta;\Od,a^\delta,\Cvr)$ such that
\[
\beta(\delta) F_\Cvr(a^\delta,\cdot)\to f_\Cvr(a,\cdot),\quad \delta\to 0
\]
uniformly on compact subsets of $\Omega$.
\end{theorem}

This convergence also holds true up to the ``nice'' parts of the boundary; moreover, considering ratios of observables corresponding to different $\varpi$'s, one can get rid of normalization issues. We work this out in Theorem \ref{thm: mainconv}.

A striking feature of our new observables is their direct relation to spin correlations. Let $\Od$ be a \emph{simply connected} domain with $m$ punctures, that is, $m$ single faces $f_1,\dots,f_m$ removed, and let $\Cvr$ be the cover that branches around each of these punctures. Then, it turns out that $F_\Cvr(a,b)$, $b\in\partial\dOd$, is (up to a fixed complex factor, see Proposition \ref{prop: obs_corr}) equal to
\[
\rZ_{ab}\E_{ab}[\sigma(f_1)\dots\sigma(f_m)],
\]
where $\rZ_{ab}$ and $\E_{ab}$ stand for the partition function and the expectation for the Ising model with Dobrushin boundary conditions: ``$-$'' on the
$(a b)$ boundary arc and ``$+$'' on $(b a)$. This, together with convergence results for the observables, gives the following corollary (with the notation ``$\E_+$'' referring to ``$+$'' boundary conditions everywhere on $\pa\Od$):

\begin{corollary}
\label{cor: intRatio2}
Let $(\Od,a^\mesh,b^\mesh)$ approximate $(\Omega,a,b)$ as $\delta\to0$. Then
 \begin{equation}
\label{intRatio2} \frac{\E_{a^\mesh b^\mesh} [\,\sigma(z_1^\mesh)\dots\sigma(z_m^\mesh)\,]}
{\E_+[\,\sigma(z_1^\mesh)\dots\sigma(z_m^\mesh)\,]}\ \mathop{\rightarrow}\limits_{\mesh\to 0}\
\vartheta(\phi(z_1),\dots,\phi(z_m)),
\end{equation}
where $\vartheta=\vartheta^{\C_+}_{\infty,0}$ are explicit functions and $\phi$ is a conformal map
from $\Omega$ onto the upper half-plane $\C_+$ sending $a$ to $\infty$ and $b$ to $0$.
\end{corollary}}

In Section~\ref{Sec: explicit_half_plane} we give explicit formulae for $\vartheta$ in $\C_+$, and hence, by conformal invariance, for all simply connected domains. For example,
\[
\vartheta^{\Omega}_{ab}(z)=\cos\,[\pi\text{hm}_{\Omega}(z,(ab))]\,,
\]
where $\text{hm}_{\Omega}(z,(ab))$ stands for the harmonic measure of the arc $(ab)$ in $\Omega$ as viewed from $z$. These formulae for $m=1,2$ were previously conjectured by means of Conformal Field Theory, see \cite{BurkhardtGuim} and earlier papers. To the best of our knowledge, the explicit formulae for $m\geq 3$ are new.

Corollary \ref{cor: intRatio2} admits a number of generalizations. Let $\Od$ approximate a finitely connected domain $\Omega$ with $k$ inner boundary components $\gamma_1,\dots,\gamma_k$ (possibly macroscopic). Then, for any $m\leq k$, one has
\begin{equation}
\label{intRatio1} \frac{\E_{a^\mesh b^\mesh}
[\,\sigma(\gamma_1^\mesh)\dots\sigma(\gamma_m^\mesh)\,]}
{\E_+[\,\sigma(\gamma_1^\mesh)\dots\sigma(\gamma_m^\mesh)\,]}\
\mathop{\rightarrow}\limits_{\mesh\to 0}\
\vartheta_{ab}^{\Omega}(\gamma_1,\dots,\gamma_m)\,,
\end{equation}
where the functions $\vartheta_{ab}^{\Omega}(\gamma_1,\dots,\gamma_m)$ are conformally invariant, the expectations $\E_{ab}$, $\E_+$ are taken for the Ising model with Dobrushin (respectively, ``+'') boundary conditions on the outer boundary component and monochromatic on inner components $\gamma_j$, meaning that we constrain the spins to be the same along each component, but do not specify a priory whether it is plus of minus. In this case, $\sigma(\gamma_j^{\delta})$ denotes the (random) spin of the component $\gamma_j$.

Further, closely following the route proposed by Hongler in~\cite{HThesis}, we prove a Pfaffian formula which generalizes (\ref{intRatio1}) to the case of $2n$ boundary change operators (in other words, ``$+/-/\dots/+/-$'' boundary conditions with $2n$ marked boundary points,
see Section~\ref{Sec: pfaffian}). For $m=1,2$ this Pfaffian formula (along with the expressions for $\vartheta$) was previously derived by means of
Conformal Field Theory \cite{BurkhardtGuim}, whereas we give it a rigorous proof for general $m$ both in discrete, and, thanks
to the convergence theorem, in continuous settings.

Another application of our new observables~\cite{IzInt} is the proof of convergence of (multiple) Ising interfaces to SLE curves in multiply connected domains. In that context, a proper choice of the observable $F_\varpi$ (i.e., the corresponding double cover $\varpi$) guarantees its martingale property with respect to the growing interface. To prove that property, it is important to relate the values of $F_\varpi$ to the partitions function of the model with relevant boundary conditions. In Section \ref{Sec: pfaffian}, we show how to do it in the most general case, see Proposition~\ref{prop: ObsZ}. The simplest example of an SLE process treated in this way (for which the use of a non-trivial double cover is essential) is a radial Ising interface converging to radial SLE${}_3$.

For simplicity, in the present paper we work on the square grid, but all our proofs remain valid for the self-dual Ising model defined on
isoradial graphs (e.g., see \cite{CHS2}).  We refer the reader interested in a detailed presentation
of the basic notions of discrete complex analysis on those graphs to the paper \cite{CHS} and
the reader interested in the history of the Ising model to the paper \cite{CHS2} and references
therein.

\subsection{Organization of the paper}
In Section \ref{Sec: notations}, we fix the notations and conventions regarding discrete domains and the Ising model. In
Section~\ref{Sec: observables_limits}, we give the definition of the spinor observable and discuss its properties (in particular, discrete holomorphicity and boundary conditions), as well as the connections to spin correlations. We then define the continuous counterparts of the observables and briefly discuss their properties. Section \ref{Sec: proof} is devoted to the proof of main convergence results for spinor observables: Theorem~\ref{thm: mainconv_int_1} (convergence in the bulk) and Theorem~\ref{thm: mainconv} (convergence on the boundary). We generalize our results to the case of multiple boundary change operators in Section~\ref{Sec: pfaffian}. Finally, in Section~\ref{Sec: explicit_half_plane} we give explicit formulae for the continuous observables $f_\Cvr$ in the punctured half-plane and for the scaling limits $\vartheta^{\C_+}_{\infty,0}$ appearing in Corollary \ref{cor: intRatio2}.

\subsection*{Acknowledgments}
We would like to thank Stanislav Smirnov who involved us into the subject of the critical
planar Ising model for many fruitful discussions. We are also grateful to Cl\'ement Hongler for many
helpful comments and remarks. Some parts of this paper were written at the IH\'ES,
Bures-sur-Yvette, and the CRM, Bellaterra. The authors are grateful to these research centers
for hospitality.

This research was supported by the Chebyshev Laboratory (Department of Mathematics and
Mechanics, Saint-Petersburg State University) under the Russian Federation Government grant
11.G34.31.0026. The first author was partly funded by P.Deligne's 2004 Balzan prize in
Mathematics (research scholarship in 2009--2011) and by the grant \mbox{MK-7656.2010.1}. The second author was supported by the
European Research Council AG CONFRA and the Swiss National Science Foundation.

\section{Notation and conventions}
\setcounter{equation}{0}

\label{Sec: notations}

\subsection{Graph notation.} By \emph{(bounded) discrete domain} (of mesh $\delta$) $\Od$ we mean a (bounded) \emph{connected}
subset of the square lattice $\delta \Z^2$ (an example of a discrete domain is given on
Fig.~\ref{Fig: grid}). More precisely, a discrete domain is specified by three sets
$\mathcal{V}(\Omega^\delta)$ (vertices), $\mathcal{F}(\Od)$ (faces) and $\mathcal{E}(\Od)=\mathcal{E}_{\mathrm{in}}(\Od)\cup
\mathcal{E}_{\mathrm{bd}}(\Od)$ (interior edges and boundary half-edges, respectively), with the
following requirements:
\begin{itemize}
\item all four edges and vertices incident to any face $f\in \mathcal{F}(\Od)$ belong to $\mathcal{E}(\Od)$;
\item every vertex in $\mathcal{V}(\Od)$ is incident to four edges or half-edges in $\mathcal{E}(\Od)$;
\item every vertex that is incident to at least one edge or half-edge $e\in \mathcal{E}(\Od)$ belongs to
$\mathcal{V}(\Od)$;
\item at least one of two faces incident to any edge $e\in \mathcal{E}_{\mathrm{in}}(\Od)$ belongs to $\mathcal{F}(\Od)$.
\end{itemize}
For an interior edge $e\in \mathcal{E}_{\mathrm{in}}({\Od})$ we denote by $z_e$ its midpoint. For a
boundary half-edge $e\in \mathcal{E}_{\mathrm{bd}}(\Od)$ we denote by $z_e$ its endpoint which is not a
vertex of $\Od$. When no confusion arises we will identify an edge (or half-edge) $e$
with a point $z_e$.
%Note that all those $z_e$ belong to the lattice $\DS$ which is again a
%(rotated) square grid.

By the \emph{boundary} $\partial \Od$ of $\Od$ we will mean the set of all its boundary
half-edges $\mathcal{E}_{\mathrm{bd}}(\Od)$ or, if no confusion arises, the set of corresponding
endpoints $z_e$.

\smallskip

A \emph{double cover} of a discrete domain $\Od$ is a
graph $\dOd$ with a two-to-one local graph isomorphism $\Cvr:\dOd\rightarrow \Od$. Given a
marked boundary half-edge $a\in \partial \Od$, one can describe points $z$ on a double
cover by lattice paths $\gamma$ running from $a$ to $z$ in $\Od$, modulo homotopy and modulo
an appropriate subgroup of the fundamental group.
If $\Od$ is $(k\!+\!1)$-connected, that is, has $k$ holes, then there are $2^k$ double covers, including the trivial one. Namely, to define a cover, for each hole one has to specify whether a loop surrounding this hole lifts to a loop in the double-cover, or to a path connecting points on different sheets. In the latter case we will say that $\Cvr$ \emph{branches} around that hole. If $z$ is a point on a double cover $\dOd$, we let $z^{\ast}\in\dOd$ be defined by
$\Cvr(z^{\ast})=\Cvr(z)$ and $z^{\ast}\neq z$. We will also use the obvious notation $\mathcal{V}(\dOd)$,
$\mathcal{E}(\dOd)$ etc.

\smallskip

\subsection{Ising model notation.} We will work with the low-temperature contour representation of the critical Ising model in
$\Od$ (see~\cite{Palmer07}). We call a subset $S$ of edges and half-edges in $\Od$ (see
Fig.~\ref{Fig: grid}, note that we admit inner half-edges in $S$) a \emph{generalized
configuration} or a \emph{generalized interfaces picture} for this model, if
\begin{itemize}
\item each vertex in $\Od$ is incident to $0,2$ or $4$ edges and half-edges in $S$;
\item if an edge $e=e'\cup e''$ consists of two halves $e',e''$, then at most one of those three
belongs to $S$.
\end{itemize}
We will denote the set of all generalized configurations in $\Od$ by $\Conf_\mathrm{gen}(\Od)$. By the
\emph{boundary} $\partial S$ of $S\in\Conf_\mathrm{gen}(\Od)$ we will mean the set of all half-edges $e\in
S$ or corresponding points $z_e$, if no confusion arises. The \emph{partition function}
of the critical Ising model is given by
\begin{equation}
\label{Zfree} \rZ(\Od)\ =\sum\limits_{S\in \Conf(\Od)} x^{|S|}, \qquad
x=x_{\mathrm{crit}}=\sqrt{2}-1
\end{equation}
(the value $x=x_{\mathrm{crit}}$ will be fixed throughout the paper). Here and below $|S|$ is the total number of edges and half-edges in $S$, and
\[
\Conf(\Od):=\{S\in \Conf_\mathrm{gen}(\Od):\partial S\subset\partial \Od\}.
\]
The formula (\ref{Zfree})
endows the set $\Conf(\Od)$ of configurations, which corresponds to free boundary conditions in the spin representation, with a probability measure, the
probability of a particular configuration $S$ being ${x^{|S|}}/{\rZ(\Od)}$.

We will mostly work with subsets of $\Conf(\Od)$, and restrictions of the probability measure to those subsets. Thus, we denote
\begin{equation}
\label{defconf}
\begin{aligned}
&\Conf_{+}(\Od) := \{S\in \Conf(\Od):\partial S=\emptyset\},\\
&\Conf_{e_1,\dots,e_{n}}(\Od):= \{S\in \Conf(\Od):\partial S=\{z_{e_1},\dots z_{e_n}\} \;\mod\; 2\},
\end{aligned}
\end{equation}
where ``mod 2'' means that if some of $e_1,\dots,e_n$ appear several
times in the subscript (it will be useful for us to allow this), we keep in $\partial S$ only those appearing an odd number of times. In the spin
representation, the subset $\Conf_+(\Od)$ corresponds to \emph{locally monochromatic} boundary
conditions, that is, along each boundary component the spins are required to be the same
(although they may be different on different components). If $a_1,\dots,a_{2n}\in \partial
\Od$, then $\Conf_{a_1,\dots,a_{2n}}(\Od)$ corresponds to the configurations where the spins
change from ``$+$'' to ``$-$'' at the boundary points (half-edges) $a_1,\dots,a_{2n}$.

\begin{rem}
To simplify the notation, we will write $\Conf_{e_1,\dots,e_{n}}(\Od)$ instead of $\Conf_{\Cvr(e_1),\dots,\Cvr(e_{n})}(\Od)$ when $e_1,\dots,e_n\in\dOd$. One should remember that we always consider Ising configurations or generalized interfaces pictures in the planar domain $\Od$ itself, even though we will define observables on double covers $\dOd$.
\end{rem}

\section{Spinor holomorphic observables and their limits}
\setcounter{equation}{0}

\label{Sec: observables_limits}

\begin{figure}[t]
\includegraphics[width=\textwidth]{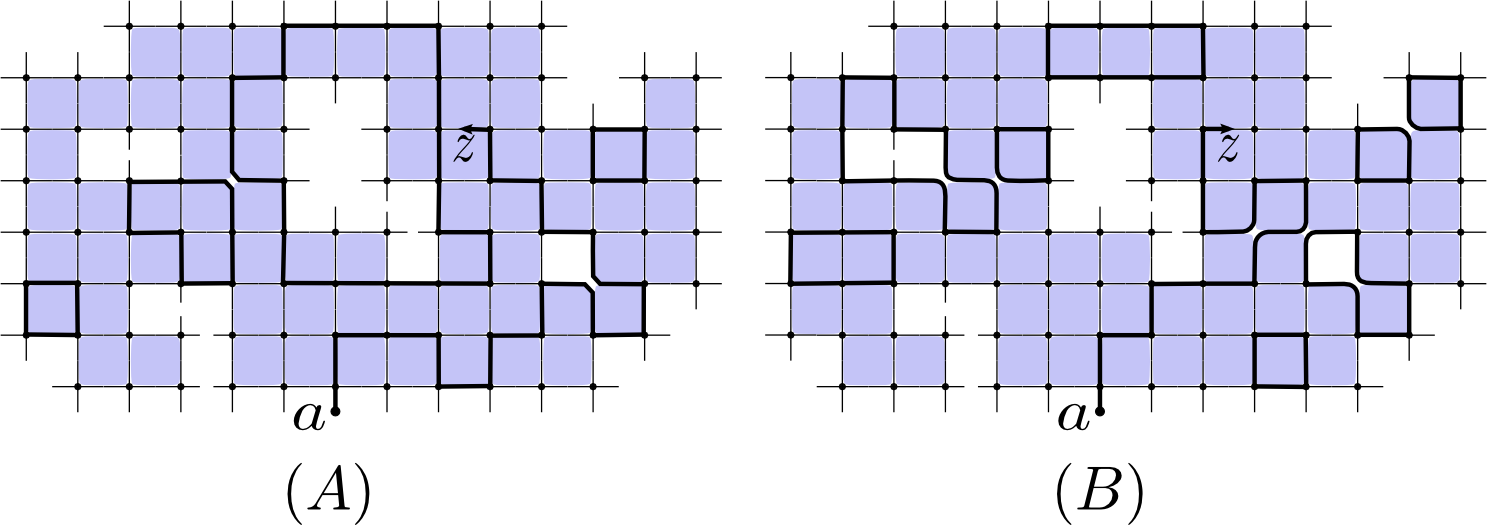}
\caption{\label{Fig: grid} An example of four-connected discrete domain $\Od$ and two generalized interfaces pictures
$S_{A,B}\in\Conf_{a,z}(\Od)$, each decomposed into a collection of loops and a simple lattice path
$\gamma_{A,B}:a\rightsquigarrow z$, as required in Definition~\ref{MainDef}. For a double cover $\dOd$, branching, say, around two small holes but not around the big central one, one has $l(S_A)=0$, $l(S_B)=2$, and $s(z,\gamma_A)=-s(z,\gamma_B)$ as, being lifted on $\dOd$, $\gamma_A$ and $\gamma_B$ end on different sheets.}

\end{figure}

\subsection{Discrete holomophic spinor observables.} \label{subsec: discr_hol_spinors}
In this subsection we will construct spinor observables and prove their discrete holomorphicity.
These observables should be considered as natural generalizations of fermionic observables
introduced by Smirnov~\cite{Smirnov06,CHS2} to the multiply connected setup. A
discrete domain $\Od$, its double cover $\Cvr: \dOd\rightarrow \Od$, and a boundary half-edge $a\in\partial\dOd$ will be fixed throughout this
subsection. In order to give a consistent definition for all discrete domains, we need the
following (technical) notation. The half-edge $a$, oriented from an inner vertex to $z_a$, can be thought of as a complex number. Then we set
\begin{equation}
\label{DefEtaA} \eta_a:= e^{-\frac{i}{2}(\arg(a)+\frac{1}{2}\pi)}=(ia/|a|)^{-\frac{1}{2}}
\end{equation}
for some fixed choice of the sign. Note that, if $a$ is south-directed, then $\eta_a=\pm 1$.

\smallskip

Given a point $z=z_e$ (i.e., the midpoint of an edge or the endpoint of a boundary half-edge $e\in E(\dOd)$) and a configuration
$S\in \Conf_{a,e}(\Od)$, we introduce the \emph{complex phase} of $S$ with
respect to $z$. First, we decompose $S$ into a collection of loops and a path $\gamma$ running
from $\Cvr(a)$ to $\Cvr(e)$ so that there are no transversal intersections or self-intersections (see
Fig.~\ref{Fig: grid}). A loop in such a decomposition will be called \emph{non-trivial} if it does not lift to a closed loop on the double cover (that is, if it lifts to a path between points on different sheets). We denote by $l(S)$ the number of non-trivial loops in $S$. We also
introduce a sign $s(z,\gamma):=+1$, if $\gamma$ lifts to a path from $a$ to $z$ on $\dOd$, and $s(z,\gamma):=-1$, if
it lifts to a path from $a$ to $z^\ast$.

\begin{definition}
\label{MainDef} The complex phase of a configuration $S\in\Conf_{a,z}(\Od)$ with respect
to a point $z$ lying on a double cover $\Cvr:\dOd\to\Od$ is defined as
\begin{equation}
\label{CW} W_\Cvr(z,S):= e^{-\frac{i}{2}\wind(\gamma)}\cdot (-1)^{l(S)}\cdot s(z,\gamma),
\end{equation}
where $\wind(\gamma)$ denotes the winding (i.e., the increment of the argument of the tangent vector) along
$\gamma$. Then, we define a {\bf spinor observable} on the double cover $\dOd$ as
\begin{equation}
\label{Observable} F_\Cvr(a,z):= i\eta_a \cdot \!\!\sum_{S\in \Conf_{a,z}(\Od)} W_\Cvr(z,S) x^{|S|}.
\end{equation}
\end{definition}

\begin{rem}
(i) $W_\Cvr(z,S)$ does not depend on the way one chooses the decomposition of a given
configuration $S$ into loops and the path $\gamma$. The proof is elementary, and we leave it to the reader. Note that it is
sufficient to check that the second factor $(-1)^{l(S)}s(z,\gamma)$ is independent of a decomposition, as the rest is
well known (e.g., see \cite[Lemma~7]{HonSmi}).

\smallskip

\noindent (ii) By definition, $F_\Cvr(a,z^\ast)\equiv -F_\Cvr(a,z)$, thus we call
$F_\Cvr(a,z)$ a \emph{spinor}.

\smallskip

\noindent (iii) If $\Cvr$ is the trivial cover, then Definition~\ref{MainDef} reproduces the
original construction due to Smirnov (e.g., see \cite[eq.~(2.10)]{CHS2}). We denote this observable by $F_0(a,z)$ and the corresponding complex phase by
$W_0(z,S)$.
\end{rem}

The most important ``discrete'' properties of the observable (\ref{Observable}) are revealed in Theorem \ref{Thm: shol} below, which states its \emph{s-holomorphicity} (see \cite[Section~3]{Smirnov10} or
\cite[Definition~3.1]{CHS2}) and describes the boundary conditions.
We introduce the
following notation: given a vertex ${v}\in \mathcal{V}(\Od)$, we consider four nearby corners of faces
incident to $v$, and identify them with the points ${v}_k:={v}+{e^{i\pi
(2k+1)/4}}\cdot\delta/{2\sqrt{2}}$, $k=0,1,2,3$. We denote sets of all corners of $\Od$ and
its double cover $\dOd$ by $\Corn(\Od)$ and $\Corn(\dOd)$, respectively. Similarly to
(\ref{DefEtaA}), for a corner $c=v_k\in\Corn(\dOd)$ we set
\[
\eta_c:=(i(c\!-\!v)/|c\!-\!v|)^{-\frac{1}{2}}:=e^{-i\pi (2k+1)/8}
\]
(again, the particular choice of square root signs is unimportant, so we fix it once forever
for each of four possible orientations of $c\!-\!v$). We denote by
\[
\Pr_{\eta_c}(F):=\Re
(\overline{\eta_c}\,F)\,\eta_c={\textstyle\frac{1}{2}}(F+\eta_c^2\,\overline{F})
\]
the orthogonal projection of a complex number $F\in\C$ onto the line $\eta_c\R$.

\smallskip

\begin{theorem}
\label{Thm: shol} For any corner $c\in \Corn(\dOd)$ formed by edges or half-edges $z', z''\in \mathcal{E}(\dOd)$, one
has
\begin{equation}
\label{shol} \Pr_{\eta_c}(F_\Cvr(a,z'))=\Pr_{\eta_c}(F_\Cvr(a,z'')).
\end{equation}
Moreover, if $b\in\partial \dOd\setminus\{a,a^\ast\}$ is a boundary half-edge, then
$F_\Cvr(a,b)\parallel \eta_{b}$, i.e.,
\begin{equation}
\label{bc_0} \Pr_{i\eta_b}(F_\Cvr(a,b))=0.
\end{equation}
%where $\eta_b$ is defined (up to a sign which is unimportant for us) similarly to (\ref{DefEtaA}).
\end{theorem}

\begin{rem}
Since our observables satisfy $F_\Cvr(a,z^\ast)\equiv -F_\Cvr(a,z)$, the
identities (\ref{shol}) at two corners $c$, $c^\ast$ such that $\Cvr(c)=\Cvr(c^\ast)$ are
equivalent. The same is fulfilled for the boundary condition (\ref{bc_0}).
\end{rem}

\begin{proof}
Let $v$ denotes the vertex incident to both $z'$ and $z''$. There exists a natural bijection
$\Pi:\Conf_{a,z'}(\Od)\to\Conf_{a,z''}(\Od)$, provided by taking ``xor'' of a
generalized configuration $S$ with two half-edges $\Cvr(vz')$ and $\Cvr(vz'')$. The well known proof
of the theorem for the trivial cover (e.g., see \cite[Proposition~2.5]{CHS2} or
\cite[Lemma~45]{HonSmi}) assures that, for any $S\in \Conf_{a,z'}(\Od)$,
$$
\Pr_{\eta_c}(W_{0}(z',S)x^{|S|})=\Pr_{\eta_c}(W_{0}(z'',\Pi(S))x^{|\Pi(S)|}).
$$
Clearly, the same holds true with $W_{0}$ replaced by $W_\Cvr$, unless $\Pi$ changes the
number of non-trivial loops $l(S)$ or $s(z',\gamma_S)\ne s(z'',\gamma_{\Pi(S)})$. However, it is easy to see
that $\Pi$ always preserves the factor $(-1)^{l}\cdot s$: for instance, if there was a
non-trivial loop in $S$ that disappeared in $\Pi(S)$, then this loop has become a part of the
path $\gamma_{\Pi(S)}$, leading to the simultaneous change of the sign $s$.

To derive the boundary condition (\ref{bc_0}), it is sufficient to note that the winding of
any curve $\gamma$ running from $a$ to $b$ is equal to $(\arg b- (\arg a \!+\!\pi))$ modulo
$2\pi$.
\end{proof}

The next proposition relates the boundary values of $F_{\Cvr}(a,\cdot)$ to spin
correlations in the Ising model. For a given double cover $\Cvr:\dOd\to\Od$ of a $(k\!+\!1)$-connected domain $\Od$, let us fix the enumeration of inner boundary components $\gamma_1,\dots,\gamma_k$ so that
\begin{quotation}
\noindent \emph{$\Cvr$ branches around each of $\gamma_1,\dots,\gamma_m$ but not around $\gamma_{m+1},\dots,\gamma_k$.}
\end{quotation}
For simplicity, below we assume that two marked boundary points $a,b$ belong to the
\emph{outer} boundary of $\dOd$. We denote by $\rZ_{ab}$ and $\E_{ab}$ the partition function and the expectation in the Ising
model with ``$-$'' boundary conditions on the counterclockwise boundary arc
$(\Cvr(a)\Cvr(b))\subset\pa\Od$, ``$+$'' on the complementary arc $(\Cvr(b)\Cvr(a))$, and
\emph{monochromatic} on all inner boundary components $\gamma_1,\dots,\gamma_k$. Recall that these boundary conditions require the spin to be constant along each $\gamma_j$. We denote this (random) spin by $\sigma(\gamma_j)$.
\begin{rem}
If some hole $\gamma_j$ is a single face, then we do not impose any boundary condition
there and $\sigma(\gamma_j)$ is just a spin assigned to this face.
\end{rem}

\begin{proposition}
\label{prop: obs_corr} If $a,b\in\pa\dOd$ belong to the outer boundary and $\Cvr(a)\ne\Cvr(b)$, then
\begin{equation}
\label{obs_corr} F_\Cvr(a,b)= \pm\eta_b\cdot
\rZ_{ab}\E_{ab}[\,\sigma(\gamma_1)\sigma(\gamma_2)\dots\sigma(\gamma_m)\,]
\end{equation}
(the choice of sign is explained in Remark~\ref{rem: bcinversion}). In particular,  $F_0(a,b)=\pm \eta_b\cdot \rZ_{ab}$. Moreover,
\begin{equation}
\label{obs_corr_a} F_\Cvr(a,a)= i{\eta_a} \cdot
\rZ_{+}\E_{+}[\,\sigma(\gamma_1)\sigma(\gamma_2)\dots\sigma(\gamma_m)\,].
\end{equation}
\end{proposition}

\begin{rem}
\label{rem: bcinversion}
The sign $\pm$ in (\ref{obs_corr}) depends on particular choices of $\eta_a$, $\eta_b$ and the sheets of $a,b$ on $\dOd$. One way to fix it is as follows: let $b$ be such that the {counterclockwise} boundary arc $(\Cvr(a)\Cvr(b))\subset\partial\Od$ lifts to $(ab)\subset\partial\dOd$ (otherwise, consider $F_\Cvr(a,b^\ast)=-F_\Cvr(a,b)$). Then, one can replace $\pm\eta_b$ in (\ref{obs_corr}) by $-\eta_a e^{-\frac{i}{2}\wind_{ab}}$, where $\wind_{ab}$ denotes the winding of the arc $(\Cvr(a)\Cvr(b))$.
%When $b$ moves along $\pa\dOd$ counterclockwise starting at $a$, the boundary conditions changes from almost all ``$+$'' to almost all ``$-$''. Then, the simple %symmetry
%\[
%\E_{+}[\,\sigma(\gamma_1)\sigma(\gamma_2)\dots\sigma(\gamma_m)\,]=
%(-1)^m\E_{-}[\,\sigma(\gamma_1)\sigma(\gamma_2)\dots\sigma(\gamma_m)\,]
%\]
%of the Ising model corresponds to the fact whether $b$ changes a sheet after a full
%turn (odd $m$) or not (even $m$).
\end{rem}

\begin{proof}
The second identity is clear from the definition (\ref{Observable}), since each configuration
$S\in\Conf_{a,a}(\Od)=\Conf_+(\Od)$ contributes the same amount
$i\eta_a(-1)^{l(S)}x^{|S|}$ to both sides of (\ref{obs_corr_a}). To prove (\ref{obs_corr}),
note that each $S\in\Conf_{a,b}(\Od)$ contributes $\pm \eta_b x^{|S|}$ to
both sides, thus we only need to check that all the signs are the same. Given a configuration $S$,
decompose it into a path $\gamma:\Cvr(a)\rightsquigarrow \Cvr(b)$ and a collection of loops. A
loop contributes to $l(S)$ (i.e., changes the sheet in $\dOd$) if and only if it has an odd
number of components $\gamma_1,\dots,\gamma_m$ inside. Hence, removing all those loops from the configuration
results in the same sign change $(-1)^{l(S)}$ for both sides. Removing the other loops (having an even
number of those $\gamma_j$'s inside) does not change the signs. After removing all loops, moving
$\gamma$ across any of $\gamma_j$ will change $\sigma(\gamma_j)$ and the sheet on which the lifting of $\gamma$ ends (i.e., the sign
$s(b,\gamma)$), again resulting in $(-1)$ factor at both sides. If finally $\gamma$ goes along the
counterclockwise arc $(\Cvr(a)\Cvr(b))$, then $S$ contributes $x^{|S|}$ to $\rZ_{ab}\E_{ab}[\sigma(\gamma_1)\sigma(\gamma_2)\dots\sigma(\gamma_m)]$, while its contribution to the left-hand side is equal to $i\eta_a e^{-\frac{i}{2}(\wind_{ab}-\pi)} x^{|S|}=-\eta_ae^{-\frac{i}{2}\wind_{ab}}x^{|S|}=\pm\eta_b x^{|S|}$, if $\gamma$ lifted to the double cover ends at $b$, and $\eta_ae^{-\frac{i}{2}\wind_{ab}}x^{|S|}$, if it ends at $b^\ast$.
\end{proof}

\subsection{Continuous spinors and convergence results.} \label{Sec: cont_spinors} In this section we introduce continuous counterparts of the discrete holomorphic spinor observables, which we will later prove to be scaling limits thereof. For a moment, let us assume that $\Omega$ is a bounded finitely connected domain whose boundary components are single points $\gamma_1=\{w_1\},\dots,\gamma_s=\{w_s\}$ and \emph{smooth} curves $\gamma_0,\gamma_{s+1},\dots,\gamma_k$. Given a double cover
$\Cvr:\dOm\to\Omega$ and a point $a\in\partial\wt{\Omega}\backslash\{w_1,\dots,w_s\}$, we denote by $f^\Omega_\Cvr(a,\cdot):\dOm\to\C$  (or just $f_\Cvr$ for shortness) an analytic function which does not vanish identically and satisfies the following properties:
\def\reflistA{\reflist{a}}
\def\reflistB{\reflist{b}}
\def\reflistC{\reflist{c}}
\def\reflistD{\reflist{d}}
\def\reflistE{\reflist{e}}
\begin{list}{(\reflist{\alph{Listcounter}})}
{\usecounter{Listcounter} \setcounter{Listcounter}{0} \labelsep 6pt \itemindent 0pt}
\label{ListA}
\item $f_{\Cvr}(a,z)\equiv -f_\Cvr(a,z^*)$ everywhere in $\wt{\Omega}$;
\label{ListB}
\item %\label{H1}
$f_\Cvr(a,\cdot)$ is continuous up to $\partial \wt{\Omega}$ except, possibly, at the single-point boundary components and at $a$, and
satisfies Riemann boundary conditions
$$
f_\Cvr(a,z)\sqrt{in_z}\in \R, \quad z\in\partial \wt{\Omega}\setminus \{a,w_1,\dots,w_s\},
$$
where $n_z$ denotes the \emph{outer} normal to $\Omega$ at $z$;
\item %\label{H2}
\label{ListC}
for each single-point boundary component $\{w_j\}$ the following is fulfilled:\\
if $\Cvr$ branches around $w_j$, then there exists a \emph{real} constant $c_j$ such that
$$
f_\Cvr(a,z) = \frac{\sqrt{i}\,c_j}{\sqrt{z-w_j}} + O(1) \quad \text{as}\ \ z\to w_j;
$$
otherwise $f_\Cvr$ is bounded near $w_j$, and thus has a removable singularity there;
\item
\label{ListD}
in a vicinity of the point $a$, one has
$$
f_\Cvr(a,z)= \frac{\sqrt{in_a}\,c^{a}}{z-a} +O(1) \quad \text{as}\ \ z\to a
$$
for a \emph{real} constant $c^{a}$.
\end{list}

The properties (\reflistB)--(\reflistD) should be thought of as natural continuous analogues of those satisfied by $F_\Cvr(a,\cdot)$ on the discrete level. Namely, (\reflistB) corresponds to the boundary condition (\ref{bc_0}); (\reflistC) turns out to be the correct formulation of this boundary condition for microscopic holes; and (\reflistD) states that $f_\Cvr$ has the simplest possible behaviour near $a$, which roughly resembles the fact that $F_\Cvr(a,\cdot)$ fails to satisfy (\ref{bc_0}) at one boundary edge only.
%We also implicitly assumed that the sign of $\sqrt{in_a}$ is chosen and fixed once forever (this is a continuous counterpart of $\eta_a$).

\begin{rem}
Lemma~\ref{lemma: uniqueness_bis} below shows that properties (\reflistA)--(\reflistD) define the function $f_\Cvr$ uniquely up to multiplication by a real constant; moreover, $c^a\neq 0$ unless $f_\Cvr$ vanishes identically. Sometimes it is convenient to fix this constant so that
\begin{equation}
 \label{norm_h_a}
c^a=1.
\end{equation}
However, below we also work with non-smooth domains, for which $c^a$ is not well-defined; therefore, we prefer to keep the multiplicative normalization of $f_\Cvr$ unfixed.
\end{rem}

The boundary value problem (\reflistA)--(\reflistD) is not easy to analyse directly. However, the following trick relates it to a much simpler Dirichlet-type problem: given a spinor $f_\Cvr$, denote
$$
h_\Cvr(z):=\Im\int^{z} (f_\Cvr(\zeta))^2d\zeta,\quad z\in\Omega.
$$
Note that the function $(f_\Cvr(\zeta))^2=(f_\Cvr(\zeta^\ast))^2$ is analytic in $\Omega$, so $h_\Cvr$ is locally well-defined and harmonic.

\newcommand\reflisth[1]{#1$\vphantom{a}^\circ_{h}$}
\def\reflistBh{\reflisth{b}}
\def\reflistCh{\reflisth{c}}
\def\reflistDh{\reflisth{d}}

\begin{lemma}
\label{lemma: h}
A holomorphic spinor $f_\Cvr$ satisfies the conditions (\reflistB)--(\reflistD) if and only if $h_\Cvr$ is a single-valued harmonic function satisfying the following properties:
\begin{list}{(\reflisth{\alph{Listcounter}})}
{\usecounter{Listcounter} \setcounter{Listcounter}{1} \labelsep 6pt \itemindent 0pt}
\item $h_\Cvr$ is continuous up to $\partial \Omega\setminus\{a,w_1,\dots w_s\}$; moreover, $h_\Cvr\equiv\const$ and $\partial_n h_\Cvr \le 0$ on all macroscopic inner boundary components $\gamma_{s+1},\dots,\gamma_k$ and the outer boundary of $\Omega$, where $\partial_n$ stands for the outer normal derivative;
\item $h_\Cvr$ is bounded from above near single-point boundary components $\{w_j\}$;
\item $h_\Cvr$ is bounded from below near the point $a$.
\end{list}
\end{lemma}
\begin{proof}
Let $f_\Cvr$ be a holomorphic spinor %continuous up to $\partial\widetilde\Omega\setminus\{a,w_1,\dots,w_s\}$
such that (\reflistB)--(\reflistD) are fulfilled. On the smooth boundary $\partial\Omega\setminus\{a,w_1,\dots,w_s\}$, one can reformulate the boundary condition (\reflistB) as $(f_\Cvr(z))^2\cdot in_z\ge 0$, which is equivalent to say that $\partial_{in} h_\Cvr \equiv 0$ and $\partial_n h\leq 0$, as stated by (\reflistBh). Further, (\reflistC) is equivalent to say that $h(z)=c_j^2\log|z-w_j|+O(1)$ as $z\to w_j$, hence (\reflistCh) holds true. In particular, $h_\Cvr$ is single-valued in $\Omega$ as it is single-valued near each of $w_j$ and constant along each of macroscopic boundary components. Similarly, (\reflistD) can be rewritten as $h(z)= -(c^a)^2\,\Re\frac{n_a}{z-a}+O(1)$ as $z\to a$, which is equivalent to (\reflistDh).

Vice versa, if $h_\Cvr$ satisfies Dirichlet boundary conditions on smooth macroscopic boundary components, then $f_\Cvr=[\partial_yh_\Cvr+i\partial_xh_\Cvr]^{1/2}$ is continuous up to \mbox{$\pa\widetilde\Omega\setminus\{a,w_1,\dots,w_s\}$}. Then, one can easily apply the similar arguments as above to deduce (\reflistB)--(\reflistD) from (\reflistBh)--(\reflistDh).
\end{proof}

\begin{lemma}
\label{lemma: uniqueness_bis}
The holomorphic spinor $f^{\Omega}_\Cvr(a,\cdot)$ with the properties (\reflistA)--(\reflistD) above, if exists, is unique up to multiplication by a positive constant. Moreover, if $\phi:\Omega\to\Omega'$ is a conformal map, then, again up to multiplicative constants,
\begin{equation}
\label{eq: conf_inv}
f^{\Omega}_\Cvr(a,\cdot)=\left(\phi'(z)\right)^{\frac{1}{2}}f^{\Omega'}_{\Cvr'}(\phi(a),\phi(\cdot))
\end{equation}
and $h^{\Omega}_\Cvr(a,\cdot)=h^{\Omega'}_{\Cvr'}(\phi(a),\phi(\cdot))$, with $\Cvr'$ being the pushforward of $\Cvr$ by $\phi$.
\end{lemma}
\begin{proof}
Suppose that $f_1$, $f_2$ both satisfy (\reflistA)--(\reflistD). Then, one can compose a linear combination $f:=c^a_2f_1- c^a_1f_2$ with non-zero coefficients, such that $f$ satisfies \mbox{(\reflistA)--(\reflistC)} and is \emph{bounded} in a neighborhood of $a$. As above, define a harmonic function
$h(z):=\Im\int^{z} (f(\zeta))^2d\zeta$, and note that it is continous up to $\pa\Omega\setminus\{w_1,\dots,w_s\}$.

As (\reflistB) implies $\partial_n h\leq 0$ everywhere on macroscopic boundary components, the function $h$ cannot attain its maximum there. Also, (\reflistC) says that $h$ is bounded from above in a neighborhood of each $w_j$.
%it has an expansion $h(z)= c_j\log|z-w_j|+O(1)$ as $z\to w_j$ for some $c_j\ge 0$. Thus, it cannot attain the maximum at ${w_j}$ either,
Then, the maximum principle gives $h\equiv\const$ and $f\equiv 0$ everywhere in $\Omega$, i.e., $f_1$ and $f_2$ are proportional to each other.

For the second claim, it is sufficient to check the properties (\reflistA)--(\reflistD) for the right-hand side of (\ref{eq: conf_inv})
and apply uniqueness, which we leave to the reader.
\end{proof}

The existence of a non-trivial solution $f_\Cvr$ to the above boundary value problem will follow from Theorem \ref{thm: mainconv_int_1}; we also refer the reader to \cite{HonPhong} for a purely analytic techniques developed for boundary problems of this kind.

The conformal covariance property (\ref{eq: conf_inv}) immediately allows one to extend the definition of $f_\Cvr^{\Omega}$ to non-smooth of unbounded domains:

\begin{definition}
If $\Omega$ is a finitely connected bounded domain, such that $\partial \Omega$ consists of smooth arcs and single points, we define $f_\Cvr^{\Omega}(a,\cdot)$  as the unique, up to a multiplicative constant, non-zero solution to the boundary value problem (\reflistA)--(\reflistD).
Otherwise we define it by (\ref{eq: conf_inv}), where $\Omega'$ is any smooth bounded domain.

Further, we choose a harmonic function $h_\Cvr^{\Omega}(a,z):=\int^z (f_\Cvr^{\Omega}(a,\zeta))^2d\zeta$ so that $h_\Cvr^{\Omega}(a,\cdot)\equiv 0$ on the boundary component of $\Omega$ containing the point $a$, thus $h_\Cvr^\Omega$ is defined up to a multiplicative constant as well.
\end{definition}

\begin{rem}
An \emph{equivalent} definition of $h_\Cvr^\Omega(a,\cdot)$ for non-smooth domains would be to impose the conditions (\reflistBh)--(\reflistDh) given in Lemma~\ref{lemma: h} and the condition that $f_\Cvr:=[\partial_yh_\Cvr+i\partial_xh_\Cvr]^{1/2}$ is a spinor on $\widetilde\Omega$. Indeed, the only condition in in Lemma~\ref{lemma: h} that relies on smoothness of $\partial\Omega$ is $\partial_n h_\Cvr\le 0$ which can be reformulated in the following conformally invariant way:
\begin{equation}
\label{eq: d_n_h_le_0}
\text{there~is~no}~ b\in\pa\Omega\setminus\{a,w_1,\dots,w_s\}~\text{such~that}~h_\Cvr(\cdot)<h_\Cvr(b)~\text{near}~b.
\end{equation}
\end{rem}

For our convergence results, we assume that discrete domains $\Od$ approximate $\Omega$ in the Carath\'eodory topology, see \cite{Pomm} or \cite[Section 3.2]{CHS}. The reader unfamiliar with that notion can think of the (stronger) Hausdorff convergence. To simplify notation, we also assume that $\Od$ has the same topology as~$\Omega$. %, and that the marked point $a$ belongs to the outer boundary of the domain.
The first theorem says that discrete spinors defined in Section~\ref{subsec: discr_hol_spinors} (with respect to a fixed double cover $\Cvr$ of the refining domains $\Od$) are uniformly close to their continuous counterparts in the bulk of $\Omega$.

\begin{theorem}
\label{thm: mainconv_int_1}
 Suppose that $\Od$ is a sequence of discrete domains of mesh size $\delta$ approximating (in the sense of Carath\'eodory) a continuous finitely connected domain $\Omega$, and that $a^\mesh\in \pa\dOd$ tends to some $a\in\pa \widetilde{\Omega}$ which is not a single-point boundary component. Then, there exists a sequence of normalizing factors $\beta(\delta)=\beta(\delta;\Od,a^\delta,\Cvr)$ such that
$$
\beta(\delta) F_\Cvr(a^\delta,\cdot)\to f^\Omega_\Cvr(a,\cdot),\quad \delta\to 0,
$$
uniformly on compact subsets of $\Omega$.
\end{theorem}

\begin{proof}
See Section~\ref{Sec: proof}.
\end{proof}

When extending this convergence to the boundary, we will impose additional regularity assumptions:

\begin{definition} \label{def: reg_conv}
 We say that a sequence of discrete domains $\Od$ with marked boundary points $b^\mesh$ approximating a planar domain $\Omega$ with a marked boundary point $b$ is \emph{regular} at $b$, if
 \begin{itemize}
 \item near $b$, the boundary $\pa\Omega$ locally coincide with a horizontal or vertical line;
 \item there exist $s,t>0$, such that, for any $\delta$, $V(\Od)$ contains a discrete rectangle
\[
\textstyle R^\delta(s,t):=\{\delta\cdot (k+i(l\!+\!\frac{1}{2})):~-s\leq k\mesh\leq s, 0\leq
l\mesh\leq t\},
\]
shifted and rotated so that $b^{\delta}$ is the midpoint of its boundary side, and
$\partial \Od$ coincides with that side in the $s$-neighborhood of $b^{\delta}$.
 \end{itemize}
 \end{definition}

\begin{rem}
In fact, all our results can be directly extended to the case of a straight, but not necessarily vertical or horizontal boundary (cf. \cite[Theorem~5.6]{CHS2}). Some additional technicalities are required to prove this result in the full generality: note that $f_\Cvr^\Omega$ is not even continuous or bounded on the non-smooth boundary, so one is forced to work with ratios, as, e.g., in Theorem \ref{thm: mainconv} below.
\end{rem}

\begin{theorem}
\label{thm: mainconv} Let $\Cvr_1, \Cvr_2$ be two double covers of a bounded finitely connected domain $\Omega$ with two marked points $a,b$ on the outer boundary component, and let $\Od$ converge to $\Omega$ in the Carath\'eodory sense, $a^\mesh, b^\mesh$ be boundary points converging to $a,b$ as $\mesh\to 0$, and this convergence is regular at $a$ and $b$. Then,
\begin{equation}
\label{mainconv} \frac{F_{\Cvr_1}(a^\delta,b^\delta)F_{\Cvr_2}(a^\delta,a^\delta)}
{F_{\Cvr_1}(a^\delta,a^\delta)F_{\Cvr_2}(a^\delta,b^\delta)}\ \rightarrow\
\frac{f_{\Cvr_1}(a,b)} {f_{\Cvr_2}(a,b)}
\end{equation}
where both $f_{\Cvr_{1}}$ and $f_{\Cvr_2}$ are normalized by (\ref{norm_h_a}).
\end{theorem}

\begin{rem}
Formally, above we should have used different notation $a_1^\delta, a_2^\delta$ etc, to denote points lying on different double covers, with $\Cvr_1(a_1^\delta)=\Cvr_2(a_2^\delta)$. We prefer to keep a simpler notation for the shortness.
\end{rem}

\begin{proof}
See Section~\ref{Sec: proof}.
\end{proof}

In the next corollary, %we consider a sequence of domains $(\Od,a^\delta,b^\delta)$ converging to $(\Omega,a,b)$ and regular at $a^\delta$ and $b^\delta$.
let $\Cvr$ be a fixed double cover of $\Omega$, and $\gamma_1,\dots, \gamma_m$ be those inner
components of $\pa\Omega$ for which $\Cvr$ branches around $\gamma_j$. Denote by $\gamma^\mesh_j$ the corresponding components of $\partial\Od$.

\begin{corollary}
\label{cor: ratio} Under the conditions of Theorem~\ref{thm: mainconv}, as $\mesh\to 0$, one has
\begin{equation}
\label{ratio_exp_conv} \frac{\E_{a^\mesh
b^\mesh}[\, \sigma(\gamma^\delta_1)\dots \sigma(\gamma^\delta_m)\,]}
{\E_+[\,\sigma(\gamma^\delta_1)\dots \sigma(\gamma^\delta_m)\,]}\
\rightarrow\ \frac{f^{\Omega}_{\Cvr}(a,b)}
{f^{\Omega}_0(a,b)}=:\vartheta_{ab}^{\Omega}(\gamma_1,\dots,\gamma_m)\,,
\end{equation}
where both $f^\Omega_\Cvr$ and $f^{\Omega}_0$ are normalized by (\ref{norm_h_a}). The limit
$\vartheta_{ab}^{\Omega}(\gamma_1,\dots,\gamma_m)$ is a conformal invariant of $(\Omega;a,b)$.
\end{corollary}

\begin{proof}
Denote $\sigma(\Gamma):=\sigma(\gamma_1^\mesh)\sigma(\gamma_2^\mesh)\dots\sigma(\gamma_m^\mesh)$. Due
to Proposition~\ref{prop: obs_corr}, one has
\[
\frac{\E_{a^\mesh b^\mesh}[\sigma(\Gamma)]} {\E_+[\sigma(\Gamma)]} =\frac{\rZ_{a^\mesh b^\mesh}\E_{a^\mesh
b^\mesh}[\sigma(\Gamma)]}{\rZ_+\E_+[\sigma(\Gamma)]}\cdot \frac{\rZ_+}{\rZ_{a^\mesh b^\mesh}}
=\frac{F_\Cvr(a^\mesh,b^\mesh)}{F_\Cvr(a^\mesh,a^\mesh)}\cdot
\frac{F_0(a^\mesh,a^\mesh)}{F_0(a^\mesh,b^\mesh)}\,.
\]
Thus, (\ref{ratio_exp_conv}) immediately follows from Theorem~\ref{thm: mainconv}. The limit is conformally invariant due to the same conformal covariance property (\ref{eq: conf_inv}) of both $f_\Cvr^\Omega$ and $f_0^\Omega$ (observe also that the coefficients $c^a_\Cvr,c^a_0$ for both $f_\Cvr$ and $f_0$ change by the same factor $|\phi'(a)|^{-1/2}$ when applying (\ref{eq: conf_inv})).
\end{proof}

\begin{rem}
\label{rem: ratio_ab_ad} (i) Corollary~\ref{cor: ratio_many} gives a generalization of this result for the case of $2n\!+\!2$ marked boundary points and ``$+/-/\dots/+/-$'' boundary conditions.

\smallskip

\noindent (ii) Let a third point $d$ be marked on the outer boundary of $\Omega$ and
the convergence of $\Omega^\mesh$ to $\Omega$ is regular at $d$ as well. Then,
\begin{equation}
\label{ratio_ab_ad} \frac{\rZ_{a^\mesh b^\mesh}}{\rZ_{a^\mesh d^\mesh}}=
\frac{|F_0(a^\mesh,b^\mesh)|}{|F_0(a^\mesh,d^\mesh)|} \ \rightarrow\
\frac{|f_0(a,b)|}{|f_0(a,d)|}\,,
\end{equation}
and this limit is a conformal covariant of the multiply connected domain $(\Omega;a,b,d)$
(namely, it is multiplied by the factor $|\phi'(b)|^{1/2}|\phi'(d)|^{-1/2}$ when applying a
conformal map $\phi$). For simply connected $\Omega$'s, this is given by
\cite[Corollary~5.7]{CHS2}, and we give a proof for multiply connected domains in the end of
Section~\ref{Sec: proof}.
\end{rem}

\setcounter{equation}{0}
\section{Proof of the main convergence theorems}
\label{Sec: proof}

In this section we prove Theorems \ref{thm: mainconv_int_1} and \ref{thm: mainconv} following the scheme
developed in \cite{CHS2} for  simply connected domains. First of all, in order to transform
the boundary conditions (\ref{bc_0}) to the Dirichlet ones, we consider a discrete integral
$H^\mesh:=\Im \int (F^\mesh_\Cvr(z))^2d^\mesh z$. Further, we observe that, under a proper normalization, the functions $H^\mesh$ and $F^\mesh_\Cvr$ have non-trivial subsequential limits on compact subsets of $\Omega$ as $\delta\to 0$. We then show that any such subsequential limit is a solution to the boundary value problem (\reflistA)--(\reflistD), and Lemma \ref{lemma: uniqueness_bis} guarantees that all those limits are the same, concluding the proof of Theorem \ref{thm: mainconv_int_1}. Finally, we treat the behaviour of $F^\mesh_\Cvr$ near the boundary points $b$ and $a$ to prove Theorem \ref{thm: mainconv}.

For technical purposes, we extend our domain slightly: denote by $\pa \mathcal{F}(\Od)$ and $\pa \mathcal{V}(\Od)$
the subsets of faces and vertices that are adjacent but do not belong to $\mathcal{F}(\Od)$ and
$\mathcal{V}(\Od)$, respectively. More precisely, $\pa \mathcal{V}(\Od)$ can be identified with the set
$\mathcal{E}_{\mathrm{bd}}(\Od)$ of boundary half-edges and should be formally considered as a set of
pairs $\{(v;e):e=(v_{\mathrm{int}}v), v\not\in \mathcal{V}(\Od), v_{\mathrm{int}}\in \mathcal{V}(\Od)\}$ (e.g., see
\cite[Section 2.1]{CHS}), and $\pa \mathcal{F}(\Od)$ should be treated in the same way. We set
\[
\overline{\mathcal{V}}(\Od):=\mathcal{V}(\Od)\cup \pa \mathcal{V}(\Od)\quad\text{and}\quad \overline{\mathcal{F}}(\Od):=\mathcal{F}(\Od)\cup \pa
\mathcal{F}(\Od).
\]

We work with the discrete spinor $F(z)=F_\Cvr(a;z)$ defined on a double cover $\dOd$ of a
discrete multiply connected domain $\Od$, which we do not include in the notation unless
needed. Recall that it is s-holomorphic (i.e., satisfies (\ref{shol})) and obeys the boundary conditions
(\ref{bc_0})
%\begin{equation}
%\label{bc}
%F(e)\parallel i\eta_e,
%\end{equation}
at all boundary half-edges $e$, except for one edge $a$ on the boundary. We denote the corresponding vertex of $\pa \mathcal{V}(\Od)$ by $v_a$. Observe also that $F_\Cvr$ is not identically zero, since the positivity of spin correlations and Proposition~\ref{prop: obs_corr} yield $(i\eta_a)^{-1} F_\Cvr(a;a)>0$. These are the only properties of $F_\Cvr(a;\cdot)$ that we will use in this section.

Recall that, in the continuous case, it was proved to be useful to transform the boundary value problem (\reflistA)--(\reflistD) into a Dirichlet-type one by integrating $f_{\Cvr}^2$. The extension of this construction to the discrete setting is delicate, because the square of discrete analytic function need not be discrete analytic. Fortunately, the tools to treate this issue have already been developed in \cite{Smirnov10, CHS2}. Proposition~\ref{prop: propertiesH} below summarizes these tools. Namely, properties (1)--(3) thereof show that one can define the discrete analog of $\Im \int f_{\Cvr}^2$ as a pair of functions $H_\bullet$
and $H_\circ$, one of which is subharmonic and another superharmonic; properties (4)--(6) handle the boundary conditions, and properties (7),(8) show that these two functions cannot be too far from each other. This allows one to work with a pair $H_{\bullet},H_{\circ}$ as if it was a single harmonic function. Essentially, our analysis is based on a priori bounds for $H$ derived from simple harmonic measure estimates combined with the uniqueness of solution to the boundary value problem (\reflistA)--(\reflistD).

We define two functions $H_\bullet$
and $H_\circ$ on $\overline{\mathcal{V}}(\Od)$ and $\overline{\mathcal{F}}(\Od)$, respectively, by the following
rule: if $c\in \Corn(\Od)$ and $e\in \mathcal{E}(\Od)$, $v\in \overline{\mathcal{V}}(\Od)$, $f\in
\overline{\mathcal{F}}(\Od)$ are all incident to $c$, then
\begin{equation}
\label{defH_basic}
H_{\bullet}(v)-H_{\circ}(f):=\sqrt{2}\delta\cdot|\Pr_{\eta_c}(F(\Cvr^{-1}(e)))|^2.
\end{equation}

Thanks to the square, this definition does not depend on the choice of the sheet, and thanks
to the basic definition (\ref{shol}) of s-holomorphicity, it does not depend on the choice of
$e$ between the two edges (or boundary half-edges) adjacent to $c$.

\begin{proposition}
\label{prop: propertiesH} The functions $H_\bullet$ and $H_\circ$ obey the following
properties:

\begin{list}{(\arabic{Listcounter})} {\usecounter{Listcounter}
\setcounter{Listcounter}{0} \labelsep 6pt \itemindent -\parindent}

\item they are well-defined up to an additive constant;
\label{H: welldef}
\item \label{defH}
if $e\in \mathcal{E}(\Od)$ is incident to $f,f'\in \overline{\mathcal{F}}(\Od)$ and $v,v'\in \overline{\mathcal{V}}(\Od)$,
then:
\begin{align*}
H_\bullet(v)-H_\bullet(v')&=\Im[ (F(\Cvr^{-1}(e))^2(v-v')]; \\
H_\circ(f)-H_\circ(f')&=\Im[ (F(\Cvr^{-1}(e))^2(f-f')];
\end{align*}
\item  \label{H: subharm}
$\Delta H_{\circ} (v)\leq 0$ and $\Delta H_{\bullet}(f)\geq 0$ for all $v\in \mathcal{V}(\Od)$, $f\in
\mathcal{F}(\Od)$, where $\Delta$ stands for the discrete Laplacian
\[
(\Delta H)(x):=\frac{1}{4\mesh^2}\sum_{x_k\sim x} (H(x_k)\!-\!H(x));
\]
\item \label{H: normal}
if $v\in \pa \mathcal{V}(\Od)\setminus\{v_a\}$, then $H_{\bullet}(v_{\mathrm{int}})-H_{\bullet}(v)\geq 0$;
\item \label{H: boundary}
$H_\circ\equiv \text{const}=:C_j$ along each component of $\pa \mathcal{F}(\Od)$ (then we fix an
additive constant in the definition of $H$ so that $H_\circ\equiv 0$ on the boundary component that contains $a$);
\item \label{H: modify}
one can modify the discrete Laplacian at all vertices $v_{\mathrm{int}}\in \mathcal{V}(\Od)$ incident to
$v\in \pa \mathcal{V}(\Od)\setminus\{v_a\}$ and set values of $H_{\bullet}$ on $\partial \mathcal{V}(\Od)\setminus\{v_a\}$ to be equal to the corresponding $C_j$, so that
(\ref{H: subharm}) and (\ref{H: normal}) will still hold (moreover, this modification do not destroy any estimates or convergence results for discrete harmonic functions defined on $\mathcal{V}(\Od)$);
\item \label{H: harn1}
if $v_1,v_2,v_3,v_4\in \mathcal{V}(\Od)$ and $f_1,f_2,f_3,f_4\in \overline{\mathcal{F}}(\Od)$ are adjacent to some
inner face $f\in \mathcal{F}(\Od)$ and $m:=\min H_{\circ}(f_j)$, then, for all $j=1,2,3,4$,
\[
H_{\bullet}(v_j)-H_{\circ}(f)\leq \const\cdot(H_\circ(f)-m)
\]
with some universal constant;
\item \label{H: harn2}
if $f_1,f_2,f_3,f_4\in \overline{\mathcal{F}}(\Od)$ and $v_1,v_2,v_3,v_4\in\overline{\mathcal{V}}(\Od)$ are
adjacent to some inner vertex $v\in \mathcal{V}(\Od)$ and $M=\max H_{\bullet}(v_j)$, then, for all $j=1,2,3,4$,
\[
H_{\bullet}(v)-H_{\circ}(f_j)\leq \const\cdot (M-H_\bullet(v))
\]
with some universal constant.
\end{list}

\end{proposition}

\begin{proof}
All these properties are known in the simply connected case (e.g., see \cite[Section~3]{CHS2}).
Since (\ref{defH}), (\ref{H: subharm}), (\ref{H: harn1}) and (\ref{H: harn2}) are local
consequences of s-holomorphicity (\ref{shol}), they extend immediately to the multiply
connected setup. Properties (\ref{H: normal}) and (\ref{H: boundary}) follow from the boundary
condition (\ref{bc_0}) using (\ref{defH}). The property (\ref{H: modify}) is also a local
property of s-holomorphic function satisfying the boundary condition (\ref{bc_0}), see
\cite[Section~3.6]{CHS2} or \cite[Proof of Proposition~8]{DHNolin}. So we only need to
check (\ref{H: welldef}), i.e. that the summation of (\ref{defH_basic}) along any loop gives
zero. For homotopically trivial loops this follows from the local consistency of the
definition (\ref{defH_basic}) exactly as in the simply connected case, and extends to
loops running around holes by (\ref{H: boundary}).
\end{proof}

%Proposition \ref{prop: propertiesH} provides a toolbox that allows to work with the funciton $H^\delta$ almost as if it was discrete harmonic (which it is not, in contrast to the continious case).

\begin{rem}
\label{remark: homogeneous} Note that we have an immediate corollary of properties (\ref{H:
welldef})--(\ref{H: normal}): if $F$ is an s-holomorphic spinor satisfying the boundary
condition (\ref{bc_0}) \emph{everywhere} on $\pa \Od$ including the point $a$, then $F\equiv
0$. Indeed, (\ref{H: subharm}) implies that the corresponding function $H_\bullet$ attains its
maximal value on the boundary, and then $H_\bullet\equiv \text{const}$ due to the property
(\ref{H: normal}) which now holds true everywhere on $\pa \Od$. Moreover, similar arguments
show that a solution to the discrete boundary value problem (\ref{shol}), (\ref{bc_0}) is
unique up to a multiplicative constant (the proof mimics \cite[Remark~5.1]{CHS2}).
\end{rem}

In what follows, we denote by $\text{hm}^{\delta}_{\Od}(v,\gamma)$ the \emph{discrete
harmonic measure} of a set $\gamma\subset\overline{\mathcal{V}}(\Od)$ in the discrete domain $\Od$ viewed from a
vertex $v$. Recall that $\text{hm}^{\delta}_{\Od}(v,\gamma)$ is given by the
probability of the event that the simple random walk starting at $v$ hits $\gamma$ before
$\pa\mathcal{V}(\Od)\setminus \gamma$. We use the same notation
$\text{hm}^{\delta}_{\Od}(f,\gamma)$ for the discrete harmonic measure of a set $\gamma\subset\overline{\mathcal{F}}(\Od)$ viewed from a face $f$.
Essentially, we will use only the following elementary properties of
$\text{hm}^\mesh$, which also are fulfilled for the discrete Laplacian modified near the
boundary as it is mentioned in Proposition~\ref{prop: propertiesH}, property~(\ref{H: modify})
(see \cite[Section~3.6]{CHS2} or \cite[Proof of Proposition~8]{DHNolin}):
\begin{itemize}
\item
\emph{weak Beurling-type estimate:} there exist absolute (i.e., independent of $\mesh$, $\Od$
and $\gamma$) constants $C,p>0$ such that the following estimate holds true:
\begin{equation}
\label{Beurling} \text{hm}^{\mesh}_{\Od}(z,\pa\Od\setminus\gamma)\le C\cdot
\biggl[\frac{\dist(z;\gamma)}{\dist_{\Od}(z;\partial\Od\setminus\gamma)}\biggr]^{\!p},
\end{equation}
where $\dist_{\Od}(z,K)$ means the smallest $r>0$ such that $z$ and $K$ are connected inside $\Od\cap B_r(z)$.
\item
\emph{uniform estimates for the exit probabilities in rectangles:} if $R^\mesh=R^\mesh(2s,t)$
is a discretization of the rectangle $(-2s,2s)\times (0,t)$ with $s\ge t$,
$\gamma_0$, $\gamma_1$~denote correspondingly the bottom and the top side of $R^\mesh$, then
\begin{equation}
\label{HmRectangles} C_1\cdot\frac{\Im z}{t}\le \text{hm}^\mesh_{R^\mesh}(z,\gamma_1) \le
\text{hm}^\mesh_{R^\mesh}(z,\pa R^\mesh\setminus \gamma_0) \le C_2\cdot\frac{\Im z}{t}
\end{equation}
for all $z\in R^\mesh(s,t)$, where $C_1,C_2>0$ are some absolute constants.
\end{itemize}
Note that (\ref{Beurling}) and (\ref{HmRectangles}) hold true for arbitrary isoradial graphs (see~\cite{CHS}).

\smallskip

Now we proceed to the definition of normalizing factors $\beta(\mesh)$ used in Theorem~\ref{thm: mainconv_int_1}. Let a
sequence of discrete domains $\Od$ approximate a finitely connected planar domain $\Omega$, whose boundary $\partial\Omega$ consists of
single-point inner components $\gamma_{1}=\{w_{1}\},\dots,\gamma_s=\{w_s\}$ and continua
$\gamma_{s+1},\dots\gamma_k,\gamma_{k+1}$, where
\emph{$\gamma_{k+1}$ denotes the outer boundary of $\Omega$}.  Let $r_*>0$ be chosen sufficiently small so that, for any $r\le r_*$,
\[%\begin{equation}\label{Omega(r)_def}
\Omega(r):=\Omega\setminus [\cup_{j=1}^s {B_r(w_j)}\cup{B_{r,\Omega}(a)}]
\]%\end{equation}
has the same topological structure as $\Omega$, where $B_{r,\Omega}(a)$ denotes the proper connected component of $B_r(a)\cap\Omega$. Further, let $\mesh_*=\mesh_*(r)>0$ be chosen small enough so that, for any $\mesh\le\mesh_*$, one has $|a^\mesh-a|\le\frac{1}{2}r$ and $\gamma^\mesh_j\subset B_{\frac{1}{2}r}(w_j)$ for all $j=1,..,s$.

\smallskip

Let $F^\mesh=F^\mesh_\Cvr(a^\mesh,\,\cdot\,)$ be the spinor observable (\ref{Observable}) in $\Od$ and
$H^\mesh_\circ,H^\mesh_\bullet$ be the corresponding discrete integrals $\Im\int^\mesh (F^\mesh(z))^2d^\mesh z$ defined by
(\ref{defH_basic}) and Proposition~\ref{prop: propertiesH}. We introduce normalizing factors $\beta^\mesh_r>0$ by
\begin{equation}
\label{NormDef} (\beta^\mesh_r)^{-1}:= \max\nolimits_{\Od(r)}\left|H^\mesh\right|,
\end{equation}
where, similarly to the continuous setup, $\Omega^\mesh(r):=\Omega^\mesh\setminus [\cup_{j=1}^s {B_r^\mesh(w_j)}\cup{B^\mesh_{r,\Omega}(a)}]$ and $B^\mesh_r(w_j),B^\mesh_{r,\Omega}(a)$ stand for discrete $r$-neighborhoods of $w_j$ and $a$ in $\Omega^\mesh$.

\begin{rem} (i) Note that the right-hand side of (\ref{NormDef}) does not vanish. Indeed, if it did, $F^\mesh(a^\mesh,\cdot)$ would also vanish identically.
\smallskip \label{rem: Hnonzero}

\noindent (ii) Since the limiting function $f_\Cvr$ has singularities at $a$ and $w_j$, it is natural to cut them off when taking a maximum to get a correct normalizing factor.
\end{rem}

This choice of $\beta^\mesh_r$ guarantees that $|\beta^\mesh_r H^\mesh|\leq 1$ in $\Od(r)$. Our next goal is to show that, for a fixed $r$, the functions $\beta^\mesh_r H^\mesh$ are uniformly bounded on \emph{any} compact subset of $\Omega$. This is achieved by the following lemma:

\begin{figure}[t]
\includegraphics[width=\textwidth]{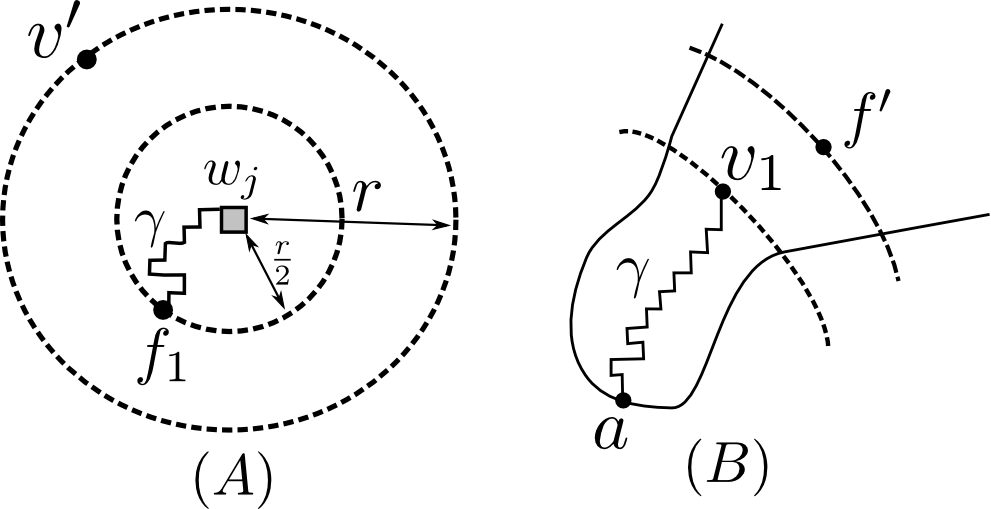}
\caption{\label{Fig: paths} (A) The proof of the estimate $\widetilde{\beta}_r^\mesh\leq \const \cdot\beta_{2r}^\mesh$ in Lemma~\ref{lemma: compareNorm}. If $-H_{\circ}$ is big at $f$, it is also big at vertices near some path $\gamma:f\rightsquigarrow w_j$. Since the harmonic measure of $\gamma$ seen from $v$ is uniformly bounded from below, $-H_\bullet(v')$ has to be big as well. (B)~Similar considerations near $a$ give the estimate $\beta_r^\mesh\leq\const\cdot \widetilde{\beta}_r^\mesh$.}

\end{figure}

\begin{lemma}
\label{lemma: compareNorm}
For any $r\le\frac{1}{2}r_*$, the ratio $\beta^\mesh_{2r}/\beta^\mesh_{r}$ is bounded uniformly in $\mesh\le\mesh_*(r)$.
\end{lemma}
\begin{proof}
Denote ${\widetilde\Omega}^\mesh(r):=\Od (r) \setminus B^\mesh_{2r,\Omega}(a)$ and
\[%\begin{equation}\label{TildeNormDef}
 (\widetilde{\beta}^\mesh_r)^{-1}:= \max\nolimits_{{\widetilde\Omega}^\mesh(r)}|H^\mesh|.
\]%\end{equation}
By definition, $\Od(2r)\subset {\widetilde\Omega}^\mesh(r)\subset \Od(r)$, hence $\beta^\mesh_{r} \le \widetilde{\beta}^\mesh_r \le \beta^\mesh_{2r}$\,. We first show that $\beta^\mesh_{2r}/\widetilde{\beta}^\mesh_{r}$ is uniformly bounded. Recall that $H^\mesh_\bullet$ is subharmonic, $H^\mesh_\circ$ is superharmonic,
$H^\mesh_\bullet(v)\geq H^\mesh_\circ(f)$ for adjacent $v$ and ${f}$, and $H^\mesh\equiv 0$ on the boundary component containing $a$. Therefore,
\begin{equation}
\label{tBeta=max(max,min)}
(\widetilde{\beta}^\mesh_r)^{-1}= \max\left\{\max\nolimits_{\pa{\widetilde\Omega}^\mesh(r)}H^\mesh_\bullet\ ;\ -\min\nolimits_{\pa{\widetilde\Omega}^\mesh(r)}H^\mesh_\circ \right\}.
\end{equation}
Note that in fact
\begin{equation}
\label{eq: maxmax}
\max\nolimits_{\pa{\widetilde\Omega}^\mesh(r)}H^\mesh_\bullet \le\max\nolimits_{\pa{\Omega}^\mesh(2r)}H^\mesh_\bullet\,.
\end{equation}
Indeed, the function $H^\mesh_\bullet$ is subharmonic in $B^\mesh_{2r}(w_j)\cap \Od$, $j=1,\dots,s$, and it cannot attain its maximum on the boundary component $\gamma_j^\mesh\subset B^\mesh_{2r}(w_j)$ because of property (4) of Proposition \ref{prop: propertiesH} (and, if $\gamma_j^\mesh$ consists of one face only, then $H^\mesh_\bullet$ is subharmonic \emph{everywhere} in $B^\mesh_{2r}(w_j)$).

Thus, either $(\widetilde{\beta}^\mesh_{r})^{-1}\le ({\beta}^\mesh_{2r})^{-1}$ and then there is nothing to prove, or
\[
(\widetilde{\beta}^\mesh_r)^{-1} = -H^\mesh_\circ(f)\quad \text{for~some}~~f\in \pa B^\mesh_r(w_j).
\]
Since $H^\mesh_\circ$ is superharmonic, in this case there exists a path of consecutive neighbors $\gamma=\{f=f_1\sim f_2\sim\dots\}$ such that $-(\widetilde{\beta}^\mesh_r)^{-1}= H^\mesh_\circ(f_{1})\ge H^\mesh_\circ(f_{2})\ge \dots$. This path can only end up at $\gamma^\mesh_j$, where the superharmonicity of $H^\mesh_\circ$ fails (see Fig.~\ref{Fig: paths}A). Denote by $\gamma'$ the set of \emph{vertices} adjacent to $\gamma$. Then, the property (8) in Proposition \ref{prop: propertiesH} ensures that
\[
H^\mesh_\bullet + (\widetilde{\beta}^\mesh_r)^{-1} \le \const\cdot ((\beta^\mesh_{2r})^{-1}-H^\mesh_\bullet) \ \ \text{everywhere~on}~\gamma',
\]
where we have used the estimate $\max\nolimits_{B^\mesh_{2r}(w_j)}H^\mesh_\bullet\le (\beta^\mesh_{2r})^{-1}$ which was explained after (\ref{eq: maxmax}). Hence, for some absolute constant $p_1>0$,
\[
H^\mesh_\bullet\le - p_1(\widetilde{\beta}^\mesh_r)^{-1} + (\beta^\mesh_{2r})^{-1} \ \ \text{everywhere~on}~\gamma'.
\]

 Further, there exists a constant $p_2>0$ independent of $r$ and $\mesh$ such that, for any vertex $v\in \pa B^\mesh_{2r}(w_j)$, one has $\text{hm}(v,\gamma'):= \text{hm}^\mesh_{\Od(2r)\cup B^\mesh_{2r}(w_j)}(v,\gamma')\ge p_2$ (see Fig.~\ref{Fig: paths}A). By subharmonicity of $H^\mesh_\bullet$, this implies
\begin{align*}
H^\mesh_\bullet(v) &\leq
(-p_1(\widetilde{\beta}^\mesh_r)^{-1}\!+(\beta^\mesh_{2r})^{-1})\cdot\text{hm}(v,\gamma')+ (\beta^\mesh_{2r})^{-1}\cdot(1\!-\!\text{hm}(v,\gamma')) \cr & \leq -p_1p_2(\widetilde{\beta}^\mesh_r)^{-1}+(\beta^\mesh_{2r})^{-1}.
\end{align*}
Since $H^\mesh_\bullet(v)\ge -(\beta^\mesh_{2r})^{-1}$, we infer that $\beta^\mesh_{2r}/\widetilde{\beta}^\mesh_r\le 2(p_1p_2)^{-1}$.

\smallskip

It remains to prove that $\widetilde{\beta}^\mesh_r/\beta^\mesh_{r}$ is uniformly bounded, which is done by exactly the same argument with the roles of $H^\mesh_\bullet$ and $H^\mesh_\circ$ interchanged, inequalities reversed, and $a$ playing the role of $w_j$. Note that the uniform lower bound for the harmonic measure $\text{hm}(f,\gamma')$ still holds (this time, $\gamma'$ will be the set of \emph{faces} adjacent to a \emph{vertex} path $\gamma$ which terminates at $a$, see Fig.~\ref{Fig: paths}B), provided that $f$ is chosen, say, to be the closest face to some fixed $z\in\Omega\backslash{B_{2r,\Omega}(a)}$ (e.g., see \cite[Lemma 3.14]{CHS}).
\end{proof}

Now we are able to claim precompactness of the families $\{H^\mesh\},\{F^\mesh_\Cvr\}$ as $\mesh\to 0$.

\begin{lemma}
\label{lemma: precompactness} Fix some sufficiently small $r>0$. Then, there exists a subsequence
\mbox{$\mesh=\mesh_k\to 0$} such that the functions $\beta^\mesh_rH^{\delta}$ converge to a
harmonic function \mbox{$H:\Omega\to\R$} uniformly on compact subsets of $\Omega$.
Moreover, the functions $\beta^\mesh_r(F^{\delta})^2$ converge to $F^2:=\partial_y H+i\partial_x H$ uniformly
on compact subsets of~$\Omega$.
\end{lemma}
\begin{proof}
Definition of $\beta^\mesh_r$ and Lemma \ref{lemma: compareNorm} guarantee that $\beta^\mesh_rH^\mesh$ are uniformly bounded in $\Od(2^{-k}r)$ for any $k\ge 0$, and hence on all compact subsets of $\Omega$. Due to \cite[Theorem~3.12]{CHS2}, the functions $\beta^\mesh_rH^\mesh$ and the spinors
$\sqrt{\beta^\mesh_r}F^{\delta}$ are thus equicontinuous on compact subsets of $\Od$, and the Arzela-Ascoli theorem implies their subsequential convergence to a continuous
function $H$ and a spinor $F$, respectively. Morera's theorem, together with the discrete holomorphicity of
$F^{\delta}$, implies that $F$ is holomorphic. Since increments of $\beta^\mesh_rH^{\delta}$ are
given by discrete integrals $\Im\int^\mesh[\beta^\mesh_r(F^{\mesh}(z))^2d^\mesh z]$, one has $H=\Im\int
(F(z))^2dz$, so $H$ is a harmonic function defined in $\Omega$.
\end{proof}

The next step is to show that all these subsequential limits solve the correct boundary value problem. It is convenient to work with the function $H$. Recall that, due to Proposition~\ref{prop: propertiesH} (properties (\ref{H: boundary}) and
(\ref{H: modify})), $H^\mesh\equiv C_j^\mesh$ on each of macroscopic boundary components $\gamma^\mesh_{s+1},\dots,\gamma^\delta_{k+1}$. By definition (\ref{NormDef}) of the normalizing factors $\beta^\mesh_r$, we have $|\beta^\mesh_r C_j^\mesh|\le 1$. Thus, taking a subsequence once more, we can assume that \begin{equation}
\label{betaCjmesh->Cj}
\beta^\mesh_r C_j^\mesh\to c_j~\text{as}~\mesh\to 0,\quad j=s\!+\!1,\dots,k\!+\!1,
\end{equation}
for some constants $c_j$ such that $|c_j|\le 1$.

\begin{lemma}
\label{lemma: identlimits}
For any subsequential limit $H$ from Lemma \ref{lemma: precompactness} satisfying (\ref{betaCjmesh->Cj}), the conditions (\reflistBh)--(\reflistDh) from Lemma~\ref{lemma: h} hold true. Moreover, $\sup\nolimits_{\Omega(r)}|H|=1$.
\end{lemma}
\begin{rem}
Below we use the equivalent reformulation (\ref{eq: d_n_h_le_0}) of the boundary condition $\partial_n H\le 0$ which does not rely on the smoothness of $\pa\Omega$.
\end{rem}
\begin{proof} \emph{Property (\reflistBh)}. Our first goal is to prove that $H$ satisfies the Dirichlet boundary conditions on all macroscopic boundary components $\gamma_j$, $j=s\!+\!1,\dots,k\!+\!1$. Fix some small $r'>0$, and recall that $|\beta^\mesh_rH^\mesh|\leq C(r')$ everywhere in $\Od(r')$, where $C(r')$ does not depend on $\mesh$. Due to superharmonicity of $H^\mesh_\circ$, for any face $f\in \Od(r')$, one has
\begin{equation}
\label{unifH} \beta^\mesh_rH^{\delta}_{\circ}(f)\geq
\beta^\mesh_rC^{\delta}_j\cdot\text{hm}^\mesh_{\Od(r')}(f,\gamma^\mesh_j) -
C(r')(1\!-\!\text{hm}^\mesh_{\Od(r')}(f,\gamma^\mesh_j))\,.
\end{equation}
Similarly, subharmonicity of $H^\mesh_\bullet$ implies that, for any vertex $v\in\Od(r')$, one has
\begin{equation}
\label{unifH1} \beta^\mesh_rH^{\delta}_\bullet(v)\leq \beta^\mesh_rC^{\delta}_j\cdot
\text{hm}^\mesh_{\Od(r')}(v,\gamma^\mesh_j) +C(r')
(1\!-\!\text{hm}^\mesh_{\Od(r')}(v,\gamma^\mesh_j))\,.
\end{equation}

Now let $f$ and $v$ approximate some point $z\in \Omega(r')$ as $\mesh\to 0$. Since both discrete harmonic measures converge to the continuous one, and \mbox{$H^\mesh_\circ(f)\le H^\mesh_\bullet(v)$} for incident $f$ and $v$, (\ref{unifH}) and (\ref{unifH1}) yield
\[%\begin{equation}\label{|H-c_j hm|}
|H(z) - c_j\,\text{hm}_{\Omega(r')}(z,\gamma_j)| \leq
C(r')(1\!-\!\text{hm}_{\Omega(r')}(z,\gamma_j))\,.
\]%\end{equation}
In particular, $H(z)\to c_j$ as $z$ tends to $\gamma_j$. The boundary condition~(\ref{H: normal}) from Proposition~\ref{prop: propertiesH} survives in the
limit and gives (\ref{eq: d_n_h_le_0}) due to \cite[Remark~6.3]{CHS2}.

\smallskip

\emph{Properties (\reflistCh) and (\reflistDh).} By subharmonicity of $H^\mesh_\bullet$, the inequality $\beta^\mesh_r H^\mesh_\bullet \leq 1$ extends from $\partial B^\mesh_r(w_j)\subset\partial\Omega^\mesh(r)$ to the whole discrete disc $B^\mesh_r(w_j)$, thus $H\le 1$ in a vicinity of $w_j$. Similarly, superharmonicity of $H^{\mesh}_\circ$ and the condition $H^\mesh_\circ \equiv 0$ on the boundary component containing $a$ imply $H\ge -1$ near $a$.

\smallskip

\emph{Normalization $\sup\nolimits_{\Omega(r)}|H|=1$.} Recall that, similarly to (\ref{tBeta=max(max,min)}), $|\beta^\mesh_rH^\mesh|$ attains its maximum $\max_{\Od(r)}|\beta^\mesh_rH^\mesh|=1$ on the boundary
\[
\partial\Od(r)\subset \pa B^\mesh_{r,\Omega}(a)\cup \pa B^\mesh_r(w_1)\cup\dots\cup \pa B^\mesh_r(w_s) \cup\gamma_{s+1}^\mesh\cup\dots\cup \gamma_{k+1}^\mesh.
\]
Note that (\ref{betaCjmesh->Cj}) and the property (\reflistBh) yield
\[
\max\nolimits_{\gamma_{s+1}^\mesh\cup\dots\cup \gamma_{k+1}^\mesh}|\beta^\mesh_rH^\mesh|\to \max\nolimits_{\gamma_{s+1}\cup\dots\cup \gamma_{k+1}}|H|.
\]
Moreover, as $\beta^\mesh_rH^\mesh$ uniformly converge to $H$ on all compact subsets of $\Omega$, one has
\[
\max\nolimits_{\pa B^\mesh_r(w_1)\cup\dots\cup \pa B^\mesh_r(w_s)}|\beta^\mesh_rH^\mesh|\to \max\nolimits_{\pa B_r(w_1)\cup\dots\cup \pa B_r(w_s)}|H|.
\]
Finally, the convergence of $\beta^\mesh_r H^\mesh$ to $H$ inside of $\Omega$ also imply
\[
\max\nolimits_{\pa B^\mesh_{r,\Omega}(a)}|\beta^\mesh_rH^\mesh|\to \max\nolimits_{\pa B_{r,\Omega}(a)}|H|,
\]
since the estimates (\ref{unifH}),(\ref{unifH1}) guarantee that $|\beta^\mesh_rH^\mesh|$ are uniformly small near the boundary component of $\Omega$ containing $a$.
\end{proof}

\begin{proof}[{\bf Proof of Theorem \ref{thm: mainconv_int_1}}]
We set $\beta(\delta):=\sqrt{\beta^\delta_r}$ for a small \emph{fixed} $r>0$. By Lemma~\ref{lemma: precompactness}, the functions $(\beta(\mesh))^2 H^\mesh$ and $\beta(\mesh)F^\mesh$ have subsequential limits $H$ and~$F$, respectively. Lemmas~\ref{lemma: identlimits},~\ref{lemma: h} and~\ref{lemma: uniqueness_bis} guarantee that all these possible limits $H$ are the same. Moreover, this unique limit is nontrivial due to the normalization condition in Lemma~\ref{lemma: identlimits}, and thus coincides with $h_\Cvr^\Omega$ normalized so that $\sup_{\Omega(r)}|h_\Cvr^\Omega|=1$. Since $H(z)=\Im \int^z(F(\zeta))^2d\zeta$, we conclude that $F=f_\Cvr^\Omega$.
\end{proof}

The next lemma shows that the convergence of $\beta(\mesh)F^\mesh$ holds true at the
boundary point $b$ (in fact, it follows from our proof that this convergence is uniform up to straight parts of the boundary $\pa\Omega$). Note that a similar result has been obtained in \cite[Theorem~5.6]{CHS2} under milder assumptions. For convenience of the reader, we give a shorter proof here using our (stronger) regularity assumptions for $\Od$ near $b$.

\begin{lemma}
\label{lemma: onboundary} Suppose that, under the conditions of Theorem \ref{thm: mainconv_int_1}, we are given a sequence of marked points $b^\delta\in\pa\Od$, $b^\delta\to b$, and the convergence $\Od\to\Omega$ is regular at $b$. Then,
$\beta(\mesh)|F^\mesh(b^\mesh)|\to |f_\Cvr^\Omega(b)|$ as $\mesh\to 0$.
\end{lemma}

\begin{proof}
Below we assume that $\Omega$ and $\Od$ are shifted and rotated %by an angle in $\{\frac{\pi}{2},\pi,\frac{3\pi}{2}\}$
so that $b^\mesh=b=0$, $\Od$ contains a discrete rectangle $R^\delta(s,t)$ for some fixed $s,t>0$, and $\partial \Od$ locally coincides with the boundary $\pa\C^\mesh_+$ of the discrete upper half-plane $\C^\mesh_+$ (see Definition~\ref{def: reg_conv}). We also assume that the additive normalization of the functions $H^\mesh$ is chosen so that they vanish on the macroscopic boundary component containing $b$.

Recall that, on compact subsets of $\Omega$, $\beta(\mesh)F^\mesh$ and $(\beta(\mesh))^2 H^\mesh$ converge to the properly normalized functions $f_\Cvr$ and $h_\Cvr$, respectively. Moreover, the uniform estimates (\ref{unifH}),(\ref{unifH1}) guarantee that the convergence $(\beta(\mesh))^2 H^\mesh\to h_\Cvr$ remains true up to $\pa\Omega\setminus\{a,w_1,\dots,w_s\}$. In particular, the functions $(\beta(\mesh))^2 H^\mesh$ converge to $h_\Cvr$ uniformly in the fixed rectangle $R(s,t)$ around $b$.

Let
\[
\mu:=(f_\Cvr(0))^2 =\pa_y h_\Cvr(0)
\]
and
\[
\begin{cases}
\cH^\mesh_\circ(f):=\Im f,& \text{for~faces}~f\in \C^\mesh_+, \cr
\cH^\mesh_\bullet(v):=\Im\, v+\frac{\mesh}{\sqrt{2}}\,, & \text{for~vertices}~v\in \C^\mesh_+
\end{cases}
\]
(note that $\cH^\mesh$ can be defined by (\ref{defH_basic}) starting with the constant s-holomorphic function $\cF^\mesh\equiv 1$ which satisfies the boundary conditions (\ref{bc_0}) on $\pa\C^\mesh_+$). Then, for any $\vare>0$, one can find a small $d>0$ such that
\[
|(\beta(\mesh))^2 H^{\delta}-\mu \cH^{\delta}|\leq \vare d\ \ \text{everywhere~~in~~} R^\mesh:=R^\mesh(2d,d)
\]
for all sufficiently small $\mesh$'s. Since both $\cH^\mesh_\circ$ and $\cH^\mesh_\bullet$ are
discrete harmonic and satisfy the same Dirichlet boundary conditions as $H^\mesh$ near $b^\mesh$, the sub- and super-harmonicity of $H^\mesh_\bullet$ and $H^\mesh_\circ$
implies that
\begin{equation}\label{H_vs_h_circ_bullet}
\begin{array}{ll}
(\beta(\mesh))^2 H^{\delta}_\circ-\mu \cH^{\delta}_\circ  \ \geq & \!\! -\vare
d\cdot \text{hm}_{R^\mesh}(\,\cdot\, ,\partial R^\mesh\setminus \pa \C^\mesh_+), \\
(\beta(\mesh))^2 H^{\delta}_\bullet-\mu \cH^{\delta}_\bullet  \  \leq & \!\! \vare
d\cdot \text{hm}_{R^\mesh}(\,\cdot\, ,\partial R^\mesh\setminus \pa \C^\mesh_+)
\end{array}\ \ \text{everywhere~in}~R^\mesh.
\end{equation}

Let $v_b:=\frac{1}{2}i\mesh$ be the inner vertex of the boundary half-edge $b$ (see notation on Fig.~\ref{Fig: proof}A). Estimating the
discrete harmonic measure $\text{hm}_{R^\mesh}(v_b,\partial R^\mesh\setminus \pa \C^\mesh_+)$ by (\ref{HmRectangles}), one obtains
\[
\sqrt{2}\mesh \cdot |\beta(\mesh)F^\mesh(b)\cos{\textstyle\frac{\pi}{8}}|^2=
H^\mesh_\bullet(v_b) \le \mu\cH^\mesh_\bullet(v_b)+\const\cdot\vare\mesh.
\]
Since $\cH^\mesh_\bullet(v_b)=\frac{\mesh}{2}+\frac{\mesh}{\sqrt{2}}=(\sqrt{2}\cos^2\!\frac{\pi}{8})\mesh$, this gives
\[
|\beta(\mesh)F^\mesh(b)|^2 \le \mu+\const\cdot \vare.
\]
Note that this bound also holds true for the value $|\beta(\mesh)F^\mesh(b')|^2$, where $b':=\mesh$
denotes the neighboring boundary edge and $v_b':=(1+\frac{1}{2}i)\mesh$ is the corresponding
inner vertex (see notation on Fig.~\ref{Fig: proof}A). Now let $b'':=(\frac{1}{2}+\frac{1}{2}i)\mesh$ be the midpoint of the edge
$(v_bv_b')$, and $f_b:=(\frac{1}{2}+i)\mesh$ be an inner face incident to both $v_b$ and
$v_b'$. Using (\ref{H_vs_h_circ_bullet}) and estimating $\text{hm}_{R^\mesh}(f_b,\partial R^\mesh\setminus \pa \C^\mesh_+)$ by (\ref{HmRectangles}),
one obtains
\[
\Re(\beta(\mesh)F^\mesh(b''))^2\cdot \mesh= H^\mesh_\circ(f_b)\ge
\mu\mesh-\const\cdot\vare\mesh.
\]
Note that for any complex number $\xi\in\C$ one has
\[
\Re(\xi^2)\le
[2\cos^2\!{\textstyle\frac{\pi}{8}}]^{-1}\cdot(|\Pr_{e^{{-i\pi}/{8}}}(\xi)|^2+|\Pr_{e^{{i\pi}/{8}}}(w)|^2).
\]
Therefore, we arrive at the inequalities
\begin{align*}
\mu-\const\cdot\vare & \leq \Re(\beta(\mesh)F^\mesh(b''))^2 \cr & \leq
[2\cos^2\!{\textstyle\frac{\pi}{8}}]^{-1}\cdot(|\Pr_{e^{{-i\pi}/{8}}}(\beta(\mesh)F^\mesh(b''))|^2+
|\Pr_{e^{{i\pi}/{8}}}(\beta(\mesh)F^\mesh(b''))|^2) \cr & =
[2\cos^2\!{\textstyle\frac{\pi}{8}}]^{-1}\cdot(|\Pr_{e^{{-i\pi}/{8}}}(\beta(\mesh) F^\mesh(b))|^2+
|\Pr_{e^{{i\pi}/{8}}}(\beta(\mesh)F^\mesh(b'))|^2) \cr &=
{\textstyle\frac{1}{2}}(|\beta(\mesh) F^\mesh(b)|^2+|\beta(\mesh) F^\mesh(b')|^2)\le
\mu+\const\cdot\vare.
\end{align*}
Since $\vare$ can be chosen arbitrary small, this yields $\lim_{\mesh\to 0}
|\beta(\mesh)F^\mesh(b)|^2=\mu$.
\end{proof}

\begin{figure}[t]
\includegraphics[width=\textwidth]{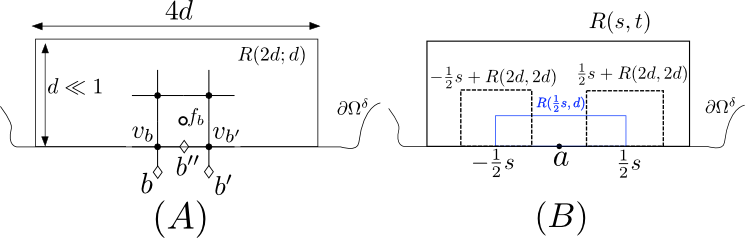}
\caption{\label{Fig: proof} (A) Notation near $b^\mesh$ used in the proof of Lemma \ref{lemma: onboundary}. (B) Notation near $a^\mesh$ used in the proof of Lemma \ref{lemma: diffCovers}.}
\end{figure}

Finally, we work out the relation of the values of discrete observables $F^\mesh$ at the points $a^\delta$ to the growth rate of their limit $f_\Cvr$ near $a$.

\begin{lemma}
\label{lemma: diffCovers}
Under the conditions of Theorem \ref{thm: mainconv_int_1}, assume also that the convergence $\Od\to\Omega$ is regular at $a$. For $j=1,2$, let $\Cvr_j$ be double covers of $\Od$, $\beta_j(\mesh)F_{\Cvr_j}^\mesh \to f_{\Cvr_j}$, and $c^{a}_j$ denote the corresponding coefficients in the expansions of $f_{\Cvr_j}$ near $a$ (see (\reflistD) in the definition of $f_{\Cvr_j}$ given in Sect.~\ref{Sec: cont_spinors}). Then,

\begin{equation}
\frac{\beta_1(\mesh)F^{\delta}_{\Cvr_1}(a^{\delta})}
{\beta_2(\mesh)F^{\delta}_{\Cvr_2}(a^{\delta})}\rightarrow
\frac{c_1^{a}}{c_2^{a}}.
\end{equation}
\end{lemma}

\begin{proof}
Below we assume that $\Omega$ and $\Od$ are shifted and rotated %by an angle in $\{\frac{\pi}{2},\pi,\frac{3\pi}{2}\}$
so that $a^\mesh=a=0$, $\Od$ contains a discrete rectangle $R^\delta(s,t)$, and $\partial \Od$ locally coincides with
$\pa\C^\mesh_+$. We also write $F^\mesh_{1,2}$ for $F^\mesh_{\Cvr_{1,2}}$ and $f_{1,2}$ for $f_{\Cvr_{1,2}}$. Denote
\[
K^\mesh:=\frac{\beta_1(\mesh)F^\mesh_1(a^\mesh)}{\beta_2(\mesh)F^\mesh_2(a^\mesh)}.
\]
Recall that both $F^\mesh_j(a)$ are positive multiples of some fixed complex number $i\eta_a$
(see Proposition~\ref{prop: obs_corr}), thus $K^\mesh>0$. Taking a subsequence, we may assume
that $K^\mesh\to k<+\infty$ as $\mesh\to 0$ (if $k=+\infty$, then swap $F^\mesh_1$ and
$F^\mesh_2$ and consider the inverse ratio $(K^\mesh)^{-1}$). Note that
\[
{\beta_{1,2}(\mesh)}\,F^\mesh_{1,2}(z)\mathop{\rightrightarrows}\limits_{\mesh\to 0}
f_{1,2}(z)=i c^a_{1,2}\,z^{-1}\!+O(1)\ \ \text{inside}~~R(s,t).
\]
Therefore, it is sufficient to prove that the function
\[
F^\mesh(z):=K^\mesh\cdot {\beta_2(\mesh)}\,F^\mesh_2(z)- {\beta_1(\mesh)}\,F^\mesh_1(z),
\]
defined in $R^\mesh(s,t)$, converges to a limit which remains bounded near $a=0$, since this
will immediately give $k c^a_2 -c^a_1= 0$ for any subsequential limit.

Note that, being a real linear combination of $F^\mesh_1$ and $F^\mesh_2$, the function $F^\mesh$ is
s-holomorphic in the discrete rectangle $R^\mesh(s,t)$ and satisfies the boundary condition
(\ref{bc_0}) on its bottom side, including the point $a^\mesh$, where $F^\mesh(a^\mesh)=0$.
Thus, (\ref{defH_basic}) allows one to define a discrete integral $H^\mesh:=\Im\int^\mesh(F^\mesh(z))^2d^\mesh z$ inside $R^\mesh(s,t)$ so that
$H^\mesh\equiv 0$ everywhere on the bottom side of $R^\mesh(s,t)$.

We claim that both functions $\beta_j(\mesh)F^\mesh_j$, and hence $F^\mesh$ and
$H^\mesh$, are uniformly bounded on the top, left and right sides of the smaller rectangle
$R^\mesh(\frac{1}{2}s,d)$ (see Fig.~\ref{Fig: proof}B for the notation), where $d\ll s$. On the top side, this follows from the uniform
convergence of $\beta_j(\mesh)F^\mesh_j$ to a continuous limit. In order to prove a uniform bound on the left and the right sides, note that, for $u\in \pm\frac{1}{2}s+R(2d,2d)$,
%\[
%\textstyle z\in \pm\frac{1}{2}s+R(2d,2d)=(\pm\frac{1}{2}s\!-\!2d,\pm\frac{1}{2}s\!+\!2d)\times
%(0,2d),
%\]
the second terms in (\ref{unifH}) and (\ref{unifH1}) can be uniformly estimated using
(\ref{HmRectangles}) in the following way:
\[
%\begin{aligned}
1-\text{hm}_{\Od(r')}(u,\gamma^\mesh_0)  \leq \const(s,t,r') \cdot \Im u.
%\end{aligned}
\]
Therefore, $(\beta_j(\mesh))^2H^\mesh_j(u)=O(\Im u)$ in neighborhoods of the left and the right sides of
$R^\mesh(\frac{1}{2}s,d)$. Hence, $\beta_j(\mesh)F^\mesh_j=O(1)$ on
those sides due to \cite[Theorem~3.12]{CHS2}.

Thus, $H^\mesh$ is uniformly bounded on the top, left and right sides of the rectangle
$R^\mesh(\frac{1}{2}s,d)$ and vanishes on its bottom side. Using super-/sub-harmonicity of
$H^\mesh$ on faces/vertices, and uniform estimates (\ref{HmRectangles}), one easily deduces
from here that $H^\mesh(u)=O(\Im u)$ everywhere in $R^\mesh({\textstyle\frac{1}{4}}s,d)$.
Applying \cite[Theorem~3.12]{CHS2} once again, we conclude that $F^\mesh=O(1)$ in
$R^\mesh({\textstyle\frac{1}{4}}s,d)$. Hence, $k c^a_2 -c^a_1= 0$.
\end{proof}

\begin{proof}[{\bf Proof of Theorem~\ref{thm: mainconv}}]
One has
\[
\frac{F_{\Cvr_1}(b^\delta)F_{\Cvr_2} (a^\delta)}
{F_{\Cvr_1} (a^\delta)F_{\Cvr_2} (b^\delta)}=
\frac{\beta_2(\mesh)F_{\Cvr_2} (a^\delta)}{\beta_1(\mesh)F_{\Cvr_1} (a^\delta)}\cdot \frac{\beta_1(\mesh)F_{\Cvr_1} (b^\delta)}{\beta_2(\mesh)F_{\Cvr_2} (b^\delta)} \to \frac{(c^a_1)^{-1}f^\Omega_{\Cvr_1}(a,b)}{(c^a_2)^{-1}f^\Omega_{\Cvr_2}(a,b)},
\]
where we have applied Lemmas~\ref{lemma: diffCovers} and~\ref{lemma: onboundary} to the first and the second factors.
\end{proof}

\begin{proof}[{\bf Proof of Remark \ref{rem: ratio_ab_ad}(ii)}]
Let $r>0$ be chosen small enough. As above, it is sufficient to prove (\ref{ratio_ab_ad}) for any
subsequence $\mesh=\mesh_k\to 0$ such that $\beta(\mesh)F^\mesh_0$ converge to some
nontrivial continuous limit. But this immediately follows from Lemma~\ref{lemma: onboundary}
applied to both boundary points $b^\mesh\to b$ and $d^\mesh\to d$, since the properly normalized observables
$\beta(\mesh)F^\mesh_0(a^\mesh,\,\cdot\,)$ converge to $f_0^\Omega(a,\,\cdot\,)$ at both $b$ and $d$.
\end{proof}

\setcounter{equation}{0}
\section{Multiple boundary change operators}
\label{Sec: pfaffian}

In this section, we follow \cite{HThesis} to extend the definition of the spinor observables to the case of \emph{multiple} marked points on the boundary. We show that these observables (we call them \emph{multi-source} ones) are still s-holomorphic. Moreover, by analysing boundary value problems they solve, we prove \emph{recurrence relations} that eventually allow one to express all these observables in terms of the basic ones introduced in Section \ref{Sec: observables_limits}.

The reason to introduce the multi-source observables is revealed in Propositions \ref{prop: obs_corr_many} and \ref{prop: ObsZ}, where we establish their relation to spin correlations and partition functions, respectively. The latter is especially important in view of two applications. First, it allows one to prove the discrete martingale property of those observables with respect to interfaces growing in multiply connected domains, leading to a description of scaling limits thereof \cite{IzInt}. Second, in the case of $2n$ microscopic holes carrying one boundary change operator each, the Kramers-Wannier duality relates the corresponding partition function to the $2n$-points spin-spin correlations in the critical Ising model with free boundary conditions. This can be used to prove the conformal covariance of their scaling limits, which would complement the results of \cite{CHHI} and this paper.

\smallskip

Consider a domain $\Od$ with $2n\!+\!1$ marked points  $a:=a_0, a_1,\dots,a_{2n}\in\partial \Od$ and an inner edge $z$.  Each configuration $S\in\Conf_{a_0,a_1,\dots,a_{2n},z}$ can be decomposed into a collection of (mutually disjoint, non-self-intersecting) loops and $n\!+\!1$ curves connecting $a_k$'s and $z$ in some manner. In order to define the complex phase of $S$, we draw $n$ artificial arcs $\nu_1,\dots,\nu_n$ connecting $\Cvr(a_1)$ to $\Cvr(a_2)$, $\dots$, $\Cvr(a_{2n-1})$ to $\Cvr(a_{2n})$, respectively, and fix the way how they lift to the double cover $\widetilde\Omega$.
Adding these arcs to a configuration promotes it to a collection of loops and a {single curve} $\gamma$ running from $\Cvr(a_0)$ to $\Cvr(z)$. As we admit intersections of the artificial arcs with curves constituting $S$, this time the loops can be self-intersecting, see Figure \ref{Fig: multi}.
\begin{definition}
\label{def: multi_source}
 We define the complex phase $W_\Cvr(z,S)$ to be $e^{-\frac{i}{2}\wind(\gamma)}(-1)^{l(S)}s(z,\gamma)$ \emph{times} $(-1)^{I(S)}$, where $I(S)$ is the number of loops in $S\cup\nu_1\cup\dots\cup\nu_n$ that have zero winding modulo $4\pi$, and other factors are as in Definition \ref{MainDef}. Further, we define the multi-source observable by the formulae
\begin{equation}
\label{Observable_MP} F_\Cvr(a_0,a_1,\dots,a_{2n};z):=i\eta_a\ \cdot\!\!\!\!\!\sum\limits_{S\in \Conf_{a,a_1,\dots a_{2n},z}(\Od)} W_\Cvr(z,S) x^{|S|}.
\end{equation}
\end{definition}
\begin{rem}
\label{rem: artificial_arcs}
% (i) We do not count self-intersections of the curve $\gamma$ or its intersections with loops. The factor $(-1)^{I(S)}$ accounts for the parity of the number of loops that have a ``wrong'' winding (equal to $0$ modulo $4\pi$).

\noindent In fact, the only data we use concerning each of the artificial arcs $\nu_s$ is its winding $\wind(\nu_s)$ modulo $4\pi$
and the way how it lifts to $\widetilde\Omega$. The reader may check that altering the choice of $\{\nu_s\}$ can only result in a sign change of the observable.
\end{rem}

%\[
%A=A_n:=\{a_1,\dots,a_{2n}\},
%\]

%\begin{equation}
%\label{Observable_manypoints}
%\begin{aligned}
%F_\Cvr(a_0,A;z) &=F_\Cvr(a_0,a_1,\dots,a_{2n};z)\cr :&= \sum\limits_{S\in
%\Conf_{\Cvr(a_0),\Cvr(a_1)\dots,\Cvr(a_{2n}),\Cvr(z)}(\Od)}\eta_{a_0}W_\Cvr(S) x^{|S|}\,.
%\end{aligned}
%\end{equation}

\begin{figure}[t]
\includegraphics[width=\textwidth]{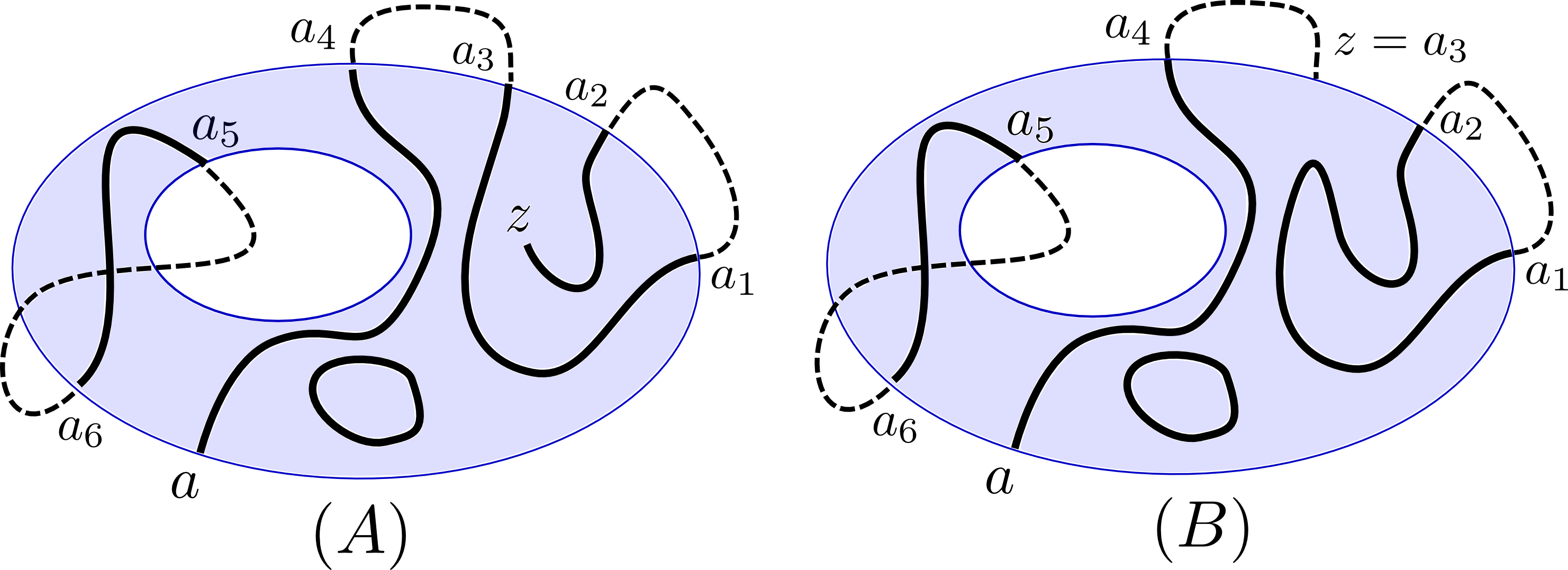}
\caption{\label{Fig: multi} (A) A doubly connected domain $\Od$ and a configuration $S\in\Conf_{a_0,a_1,\dots,a_6;z}(\Od)$. Adding artificial arcs
 (dashed lines), we promote $S$ to a collection of loops and a
simple path $\gamma$ running from $\Cvr(a_0)$ to $\Cvr(z)$. The loop containing $a_5$ and $a_6$ has winding $0$, thus contributing to $I(S)$. (B)~A configuration contributing to $f(a_0,a_1,a_2,a_3,a_4,a_5,a_6;a_3)$. It also contributes the same value (up to a complex factor $e^{-\frac{i}{2}\wind(\nu_2:a_4\to a_3)}$) to $f(a_0,a_1,a_2,a_5,a_6;a_4)$.}
\end{figure}

The following straightforward generalization of Theorem~\ref{Thm: shol} holds true:

\begin{proposition}
\label{prop: s-hol_many} The observables (\ref{Observable_MP}) are
s-holomorphic spinors satisfying the boundary condition (\ref{bc_0}) everywhere on $\partial
\dOd\setminus\{a_0,\dots,a_{2n},a_0^\ast,\dots,a_{2n}^\ast\}$.
\end{proposition}
\begin{proof}
The proof is essentially the same as the one of Theorem~\ref{Thm: shol}. The only difference is that now the bijection $\Pi$ between the sets of configurations $\Conf_{a,a_1,\dots,a_{2n},z'}$ and $\Conf_{a,a_1,\dots,a_{2n},z''}$ can create or destroy a loop which is self-intersecting. If such a loop has winding $2\pi$ modulo $4\pi$, then there is no difference with the case of a simple loop. If it has winding $0$ modulo $4\pi$, then its contribution to $e^{-\frac{i}{2}\wind(\gamma)}$ after it becomes a part of $\gamma$ is \emph{minus} that of a simple loop, which is compensated by the simultaneous change of $I(S)$ by one.
%hence, without introducing the $(-1)^{I(S)}$, factor we would have
%$$
%\Pr_{\eta_c}(W_{0}(z',S)x^{|S|})=-\Pr_{\eta_c}(W_{0}(z'',\Pi(S))x^{|\Pi(S)|})
%$$
%in this case. However, $I(S)$ also changes by one, fixing the sign in the above identity and thus proving the proposition.
The boundary conditions follow in the same way as before.
\end{proof}

We now turn to the relation of the multi-source observables with spin correlations and partition functions. Let $2n\!+\!2$ distinct boundary points $a=a_0,a_1,\dots,a_{2n}$, \mbox{$a_{2n+1}=b$} be chosen on the \emph{outer} boundary of $\dOd$, and let $\Cvr$ be a double cover that branches around boundary components $\gamma_1,\dots,\gamma_m$ and does not branch around the others. Write also $\sigma(\Gamma):=\sigma(\gamma_1)\dots\sigma(\gamma_m)$, with $\sigma(\Gamma)=1$ if $m=0$.
\begin{proposition}
\label{prop: obs_corr_many}
We have
\begin{equation}
\label{obs_corr_many}
F_\Cvr(a,A;b)= \pm \eta_b\rZ_{a,A,b}\E_{a,A,b} [\,\sigma(\Gamma)\,],
\end{equation}
where the subscripts refer to the model with boundary change operators at all the marked points $a,a_1,\dots,a_{2n},b$.
%Moreover, $F_\Cvr(a,A;a)= i\eta_a\rZ_{A}\E_{A} [\,\sigma(\Gamma)\,]$.
\end{proposition}

\begin{proof}
The proof follows from definition of $F_{\Cvr}(a,A;\cdot)$ similarly to (\ref{obs_corr}), so we leave it to the reader.
\end{proof}
\begin{rem}
\label{rem: signs}
Let us describe one way to fix the sign in (\ref{obs_corr_many}). Suppose that, starting at $a$ and tracing the outer boundary of $\dOd$ in the counterclockwise direction, we find the marked points $a,a_1,\dots,a_{2n},b$ in that order. %(and then, possibly, $a^*, a_1^*,\dots$).
Let the artificial arcs $\nu_s$ run outside the domain, following the boundary arcs $(a_{2s-1},a_{2s})$. Then, assuming that the arcs $(\Cvr(a_{2s-1})\Cvr(a_{2s}))$ carry ``$+$'' boundary conditions, one can replace $\pm\eta_b$ in the right-hand side of (\ref{obs_corr_many}) by $-\eta_ae^{-\frac{i}{2}\wind_{ab}}$, where $\wind_{ab}$ is the winding of the counterclockwise boundary arc $(\Cvr(a)\Cvr(b))$, cf.~Remark~\ref{rem: bcinversion}.
\end{rem}

For the next proposition, we allow the marked points $a=a_0,\dots,a_{2n+1}=b$ to be \emph{anywhere} on the boundary. Denote by $\Cvr_\rZ$ the double cover that branches around all boundary components of $\Od$ carrying an \emph{odd} number of marked points and does not branch around the others.
\begin{proposition}
\label{prop: ObsZ}
Let $\Od$ be a discrete domain. One has
$$
F_{\Cvr_\rZ}(a,A;b)= \pm {\eta_b}\cdot \rZ_{a,A,b}.
$$
\end{proposition}

\begin{proof}
 Any configuration $S\in\Conf_{a,A,b}$ contributes the same value $\pm {\eta_b}x^{|S|}$ to both sides of the equation, so it is sufficient to check that $W(b,S)$ does not depend on~$S$. It is convenient to add one more artificial arc $\nu_0$ connecting $b$ to $a$, and to think that the arcs $\nu_0,\dots,\nu_n$ are actually drawn on $\Od$ as simple, mutually non-intersecting curves, cf.~Remark~\ref{rem: artificial_arcs}.

By definition, an inner boundary component has an odd number of artificial arcs $\nu_s$ issuing therefrom if and only if $\Cvr_\rZ$ branches around that component. Hence, cutting along, say, the right-hand side of each $\nu_s$, we get a sheet of $\dOd$, and $\dOd$ can be obtained from two copies of that sheet by gluing along the cuts, in such a way that crossing any cut would change the sheet.

We now want to show that $W(b,S)=e^{-\frac{i}{2}\wind(\gamma)}(-1)^{l(S)+I(S)}s(b,\gamma)$ does not depend on $S$. Adding $\nu_0$ to $S\cup\nu_1\cup\dots\cup\nu_n$, we end up with a collection of loops; denote $\gamma_0:=\gamma\cup \nu_0$. Observe that a loop has winding $0$ modulo $2\pi$ (and hence contributes to $I(S)$) if and only if it has an odd number of self-intersections, and it contributes to $l(s)$ if and only if it intersects an odd number of cross-cuts. So, we can write $l(S)+I(S)=\sum_{\textrm{loops}\; \gamma_\alpha\neq \gamma_0} r(\gamma_\alpha)\;\mod \;2$, where $r(\gamma_\alpha)$ is the number of intersections of $\gamma_\alpha\setminus\{\nu_1,\dots,\nu_n\}$ with \emph{other} loops (self-intersections contribute to both $l(S)$ and $I(S)$).
Those intersections can only occur between $S$ and artificial arcs, and each intersection contributes at most once to the sum.

Further, observe that $e^{-\frac{i}{2}\wind(\gamma)}=e^{\frac{i}{2}\wind(\nu_0)}e^{-\frac{i}{2}\wind(\gamma_0)}$. Just as above, $-e^{-\frac{i}{2}\wind(\gamma_0)}$ counts the number of self-intersections of $\gamma_0$, and $s(b,\gamma)$ describes the number of intersections of $\gamma$ with the other cross-cuts. Therefore,
$$
W(S,b)=- e^{\frac{i}{2}\wind(\nu_0)}(-1)^{r(\gamma_0)+\sum_{\textrm{loops}\; \gamma_\alpha\neq\gamma_0} r(\gamma_\alpha)}.
$$
Since the exponent is the total number of intersections between the loops (not counting self-intersections), it is always even.
\end{proof}

In accordance with our convention ``$\text{mod}~2$'' in (\ref{defconf}), Definition~\ref{def: multi_source} also gives the values of $F_\Cvr$ at marked boundary points. Denote
$A:=\{a_1,\dots,a_{2n}\}$.

\begin{proposition}
\label{prop: Ind_formula} The identity
\begin{equation}
\label{Ind_Formula}
\begin{aligned}
F_\Cvr(a_0,A;z)= \sum\limits_{k=0}^{2n} \frac{F_\Cvr(a_0,A;a_k)} {F_\Cvr(a_{k};a_{k})}\,
{F_\Cvr(a_{k};z)}
\end{aligned}
\end{equation}
is fulfilled for any $z\in\dOd$.
\end{proposition}

\begin{proof}
Both sides of (\ref{Ind_Formula}) are discrete s-holomorphic spinors (defined on the same
double cover $\dOd$) satisfying the boundary condition (\ref{bc_0}) everywhere on $\pa\dOd$ except
the marked points $a_0,\dots,a_{2n}$. Moreover, for any $k=0,\dots,2n$, there is only one term
in the sum (\ref{Ind_Formula}) that fails to satisfy (\ref{bc_0}) at $a_k$. However, its value
at $a_k$ coincides with the left-hand side value $F_\Cvr(a_0,A;a_k)$. Hence, these two
\mbox{s-holo}morphic spinors are equal to each other due to Remark~\ref{remark: homogeneous},
since their difference satisfies the boundary condition (\ref{bc_0}) \emph{everywhere} on
$\pa\dOd$.
\end{proof}

Note that configurations contributing to $F_\Cvr(a_0,A;a_k)$ actually have $2n$ boundary points instead of $2n\!+\!2$. Thus, the right-hand side of (\ref{Ind_Formula}) can be expressed in terms of the similar observables with smaller number of marked points, and, recursively, in terms of the basic observables $F_\Cvr(a_k;\cdot)$. In order to do this in a convenient way, we need an additional notation. Recall that, for each $a_k$, we fix the complex number $\eta_k:=\eta_{a_k}$ according to (\ref{DefEtaA}). Then, we define the \emph{real} antisymmetric $(2n\!+\!1)\times (2n\!+\!1)$ matrix $G_\Cvr=[(G_\Cvr)_{j,k}]$ by setting, for $0\leq j<k\leq 2n$,
\begin{equation}
\label{eq: def_G}
(G_\Cvr)_{j,k}:= \frac{F_\Cvr(a_j;a_{k})}{iF_\Cvr(a_k;a_k)} = - \frac{F_\Cvr(a_k;a_{j})}{iF_\Cvr(a_j;a_j)}=:-(G_\Cvr)_{k,j}\,,
%\frac{\eta_j F_\Cvr(a_j;a_{k})}{i\eta_{k}F_\Cvr(a_j;a_j)}
\end{equation}
where we have used that $(i\eta_k)^2e^{-i\wind(\gamma_{kj})}=\eta_j^2$ for any curve $\gamma_{kj}:\varpi(a_k)\to \varpi(a_j)$ running in $\Od$.
Further, let $G_\Cvr[k_1,\dots,k_s]$ denote the sub-matrix of $G_\Cvr$ obtained by \emph{removing} rows and columns with indices $k_1,\dots,k_s$.

\begin{proposition}
 \label{prop: obs_pfaff}
We have
\begin{equation}
\label{eq: obs_pfaff}
 F_\Cvr(a_0,A;z)= \pm\sum\limits_{k=0}^{2n} (-1)^{k}\,\mathrm{Pf}\,G_\Cvr[k]\cdot F_\Cvr(a_{k};z),
\end{equation}
with the sign depending on the choices made for $\eta_k$ and $\nu_s$. In particular, the sign is ``plus'' with conventions described
in Remark~\ref{rem: Pfs_sign} below.
\end{proposition}
\begin{rem}
\label{rem: Pfs_sign}
 (i) The sign in the left-hand side of (\ref{eq: obs_pfaff}) depends on the choice of artificial arcs $\nu_s$, while for the sum in the right-hand side it depends on the sheets of $a_k$ and the signs of $\eta_{k}$, $k=1,\dots,2n$. We will assume that each $\nu_s$ lifts to a path from $a_{2s-1}$ to $a_{2s}$ on $\dOd$, and that the signs of $\eta_{k}$ are chosen so that
\begin{equation}
\label{eq: def_eta}
\exp\left[-\frac{i}{2}{\wind(\nu_s:a_{2s-1}\to a_{2s})}\right]=\frac{\eta_{2s}}{i\eta_{2s-1}}\,.
\end{equation}

\noindent (ii) One could also write (\ref{eq: obs_pfaff}) as the Pfaffian of a $(2n+2)\times(2n+2)$ matrix, obtained from $G_\Cvr$ by adding the last column with entries $F_\Cvr(a_k;z)$ and a corresponding row.
\end{rem}

\begin{proof}
We prove the claim by induction in $n$, starting with the trivial case $n=0$ and using (\ref{Ind_Formula}). For $k=1,\dots,2n$, let $k':=k+1$, if $k$ is odd, and $k':=k-1$, if $k$ is even. By definition, any configuration contributing to $F_\Cvr(a_0,A;a_{k})$ contains a curve running from $a_0$ to $a_{k'}$, appended with an artificial arc connecting $a_{k'}$ to $a_k$. Removing that arc, and taking into account (\ref{eq: def_eta}), we get the following identity:
$$
F_\Cvr(a_0,A;a_{k}) = (-1)^{k} \frac{\eta_k}{i\eta_{k'}} F_\Cvr(a_0,A[k,k'];a_{k'}),
$$
where $A[k,k']:=A\setminus\{a_{k},a_{k'}\}$. Thus, using the induction hypothesis and observing that $\eta_{k}^{-1}F(a_k,a_k)$ does not depend on $k$, we get (for $k=1,\dots,2n$)
\begin{align*}
 \frac{F_\Cvr(a_0,A;a_k)} {F_\Cvr(a_{k};a_{k})} & = (-1)^{k} \frac{\eta_k}{i\eta_{k'}} \sum\limits_{0\leq j\neq k,k'\leq 2n}(-1)^j \,\mathrm{Pf}\,G_\Cvr[j,k,k']\frac{F_\Cvr(a_j;a_{k'})}{F_\Cvr(a_{k};a_{k})}\\
&=(-1)^{k} \sum\limits_{0\leq j\neq k,k'\leq 2n}(-1)^j \,\mathrm{Pf}\,G_\Cvr[j,k,k']\frac{F_\Cvr(a_j;a_{k'})}{i F_\Cvr(a_{k'};a_{k'})}\\
&=(-1)^{k} \sum\limits_{0\leq j\neq k,k'\leq 2n}(-1)^{j+\1[{j>k'}]}\,\mathrm{Pf}\,G_\Cvr[j,k,k']\cdot(G_\Cvr)_{j,k'}\,.
\end{align*}
Due to the standard recursive formula for Pfaffians applied to the matrix $G_\Cvr[k]$, this can be written as
\begin{equation}
\label{eq: F0Ak/Fkk=}
\frac{F_\Cvr(a_0,A;a_k)} {F_\Cvr(a_{k};a_{k})}=  (-1)^{k}\,\mathrm{Pf}\,G_\Cvr[k],\quad k=1,\dots,2n.
\end{equation}
Similarly,
\[
F_\Cvr(a_0,A;a_0) = \frac{i\eta_0}{\eta_1}F_\Cvr(a_2,A[1,2];a_1)
= \frac{i\eta_0}{\eta_1}\sum\limits_{k=3}^{2n} \frac{F_\Cvr(a_2,A[1,2];a_k)} {F_\Cvr(a_{k};a_{k})}\,
{F_\Cvr(a_{k};a_{1})}.
\]
Applying (\ref{eq: F0Ak/Fkk=}) and using $\eta_0^{-1}F_\Cvr(a_0;a_0)=\eta_1^{-1}F_\Cvr(a_1;a_1)$, one arrives at
\[
\frac{F_\Cvr(a_0,A;a_0)}{F_\Cvr(a_0;a_0)}=
-\sum\limits_{k=3}^{2n} {(-1)^{k}\,\mathrm{Pf}\,G_\Cvr[0,1,k]}\cdot (G_\Cvr)_{k,1}= \mathrm{Pf}\,G_\Cvr[0].
\]
Plugging this and (\ref{eq: F0Ak/Fkk=}) into (\ref{Ind_Formula}), we obtain (\ref{eq: obs_pfaff}).
\end{proof}

\begin{corollary}
\label{cor: ratio_many}
With the notation as above and conventions of Remark~\ref{rem: signs}, let $\Od$ approximate $\Omega$ as $\mesh\to 0$, regularly at all marked points. Then,
\begin{equation}
\label{ratio_exp_conv_many} \frac{\E_{a_0^\mesh\dots
a_{2n+1}^\mesh}[\,\sigma(\Gamma)\,]}
{\E_+[\sigma(\Gamma)\,]}\,\rightarrow\,
\frac{\mathrm{Pf}\,[\,\zeta_{a_ja_k}^{-1}\vartheta_{a_ja_k}^{\Omega}(\gamma_1,\dots,\gamma_m)\,]\,_{0\le
j<k\le 2n+1}} {\mathrm{Pf}\,[\,\zeta_{a_ja_k}^{-1}\,]\,_{0\le j<k\le 2n+1}}\,,
\end{equation}
where the conformal invariants $\vartheta_{ab}^{\Omega}$ are given
by~(\ref{ratio_exp_conv}), and
$\zeta_{ab}=\zeta_{ab}^{\Omega}:=|f_0^\Omega(a,b)|^{-1}$.
\end{corollary}

\begin{proof}
Below we omit $\mesh$ for the shortness. Using the relations (\ref{obs_corr_many}) and (\ref{obs_corr_a}) for $\Cvr$ as in Proposition \ref{prop: obs_corr_many} and for the trivial cover, and then (\ref{eq: obs_pfaff}), we get
\begin{align*}
 \frac{\E_{a_0\dots a_{2n+1}}[\sigma(\Gamma)]} {\E_+[\sigma(\Gamma)]} & = \frac{F_{\Cvr}(a_0,A;{a_{2n+1}})\cdot F_{0}(a_{2n+1},a_{2n+1})}{F_{0}(a_0,A;{a_{2n+1}})\cdot F_{\Cvr}(a_{2n+1},a_{2n+1})}=\cr
& = \frac{\sum_{k=0}^{2n} (-1)^{k}\,\mathrm{Pf}\,G_{\Cvr}[k]\cdot {F_{\Cvr}(a_{k};a_{2n+1})}\cdot (i F_\Cvr(a_{2n+1};a_{2n+1}))^{-1}}{\sum_{k=0}^{2n} (-1)^{k}\,\mathrm{Pf}\,G_0[k]\cdot F_{0}(a_{k};a_{2n+1})\cdot (i F_0(a_{2n+1};a_{2n+1}))^{-1}} \cr
& = \frac{\sum_{k=0}^{2n} (-1)^{k}\,\mathrm{Pf}\,G_{\Cvr}[k]\cdot (G_\Cvr)_{k,2n+1}}{\sum_{k=0}^{2n} (-1)^{k}\,\mathrm{Pf}\,G_0[k]\cdot (G_0)_{k,2n+1}} = \frac{\mathrm{Pf}[\,(G_{\Cvr})_{j,k}\,]_{0\le j<k\le
2n+1}}{\mathrm{Pf}[\,(G_{0})_{j,k}\,]_{0\le j<k\le 2n+1}}\,,
\end{align*}
where $G_{\Cvr}$ and $G_{0}$ are given by (\ref{eq: def_G}) for corresponding double covers. This can be further rewritten as
\[
 \frac{\E_{a_0\dots a_{2n+1}}[\sigma(\Gamma)]} {\E_+[\sigma(\Gamma)]} =
\frac{\displaystyle\mathrm{Pf}\left[\frac{(G_\Cvr)_{j,k}}{(G_0)_{j,k}}\cdot
\frac{(G_0)_{j,k}}{(G_0)_{0,2n+1}}\right]_{0\le j<k\le 2n+1}}
{\displaystyle\mathrm{Pf}\left[\frac{(G_0)_{j,k}}{(G_0)_{0,2n+1}}\right]_{0\le
j<k\le 2n+1}}\,.
\]
Note that definition (\ref{eq: def_G}) and convergence (\ref{mainconv}) yield
\[
\frac{(G_\Cvr)_{j,k}}{(G_0)_{j,k}}=\frac{F_\Cvr(a_j;a_k)F_0(a_j;a_j)}{F_\Cvr(a_j;a_j)F_0(a_j;a_k)}\ \rightarrow\
\frac{f_\Cvr^\Omega(a_j,a_k)}{f_0^\Omega(a_j,a_k)}=\vartheta_{a_ja_k}^{\Omega}(\gamma_1,\dots,\gamma_m)
\]
as $\mesh\to 0$. Similarly, (\ref{eq: def_G}) and (\ref{ratio_ab_ad}) give
\[
\frac{(G_0)_{j,k}}{(G_0)_{0,2n+1}}=
%\frac{|F_0(a_j;a_k)|}{|F_0(a_j;a_j)|}\cdot\frac{|F_0(a_0;a_0)|}{|F_0(a_0;a_{2n+1})|}=
\frac{|F_0(a_j;a_k)|}{|F_0(a_0;a_{2n+1})|} \rightarrow\
\frac{|f_0^\Omega(a_j,a_k)|}{|f_0^\Omega(a_0,a_{2n+1})|}=\frac{\zeta_{a_0a_{2n+1}}}{\zeta_{a_ja_k}},
\]
and the factors $\zeta_{a_0a_{2n+1}}$ cancel out in the ratio of Pfaffians.
\end{proof}

\begin{rem}

\noindent (i) One may assume that $f_0^\Omega(a,b)$ is normalized by (\ref{norm_h_a}); in this case one has a conformal covariance rule $\zeta^{\Omega}_{a,b}=\zeta^{\phi(\Omega)}_{\phi(a),\phi(b)}|\phi'(a)|^{1/2}|\phi'(b)|^{1/2}$. The linearity of Pfaffians implies that the normalization is not important, and that the ratio (\ref{ratio_exp_conv_many}) is conformally invariant.

\smallskip

\noindent (ii) Clearly, $f^{\Omega}_0$ does not change if one adds boundary singletons inside the domain. In particular, if  $\Omega=\C_+\setminus \{w_1,\dots,w_s\}$ and $a,b\in \R$, then one has $\zeta^\Omega_{a,b}=\sqrt{\pi}\,|b-a|$ (with the normalization given by (\ref{norm_h_a})).

\smallskip

\noindent (iii) The formula (\ref{ratio_exp_conv_many}) was predicted
by means of Conformal Field Theory, see \cite[equation~(17)]{BurkhardtGuim}.
\end{rem}

\setcounter{equation}{0}
\section{Explicit computations in the half-plane}
\label{Sec: explicit_half_plane}

In this section we explicitly compute the holomorphic spinors $f_{\Cvr}^{\C_+\setminus\{w_1,\dots,w_m\}}$. By (\ref{ratio_exp_conv}), this immediately gives us the quantities
\begin{equation}
\label{ThetaW1m} \vartheta(w_1,\dots,w_m):=
\vartheta_{\infty,0}^{\C_+\setminus\{w_1,\dots,w_m\}}(w_1,\dots,w_m)
\end{equation}
and, by conformal invariance of $\vartheta$'s, all the limits
\[
\lim_{\mesh\to 0}\frac{\E_{a^\mesh b^\mesh}[\,\sigma(w_1^\mesh)\dots \sigma(w_m^\mesh)\,]}
{\E_+[\,\sigma(w_1^\mesh)\dots \sigma(w_m^\mesh)\,]}= \vartheta(\phi(w_1),\dots,\phi(w_m))\,,
\]
for discrete domains $(\Omega^\mesh;a^\mesh,b^\mesh)$ approximating an arbitrary \emph{simply
connected} domain $(\Omega;a,b)$ and inner \emph{faces} $w^\mesh_j$ tending to
$w_j\in\Omega$, where $\phi:\Omega\to \C_+$ is a conformal map such that $\phi(a)=\infty$ and
$\phi(b)=0$. Due to the Pfaffian formula (\ref{ratio_exp_conv_many}), this result easily
extends to the case of $2n\!+\!2$ marked boundary points.

It is convenient to choose $a=\infty$. Specializing the boundary value problem \mbox{(\reflistA)--(\reflistD)} and  (\ref{norm_h_a}) to the case $\Omega=\C_+$, $a=0$ and using the conformal map $z\mapsto -z^{-1}$, we get the following conditions for $f:=f_{\Cvr}^{\C_+\setminus\{w_1,\dots,w_m\}}(\infty,\cdot)$:
%(see the corresponding boundary value problem (\reflistA)--(\reflistD) for the harmonic function on page~\pageref{ListA}):

\renewcommand\reflist[1]{#1$\vphantom{a}^\circ$}
\def\reflistAC{\reflist{a}}
\def\reflistBC{\reflist{b}}
\def\reflistCC{\reflist{c}}
\def\reflistDC{\reflist{d}}
\begin{list}{(\reflist{\alph{Listcounter}})}
{\usecounter{Listcounter} \setcounter{Listcounter}{0} \labelsep 6pt \itemindent 0pt}
\item $f$ is a spinor in $\C_+\setminus\{w_1,\dots,w_m\}$ branching around each of $w_j$;
\item $f(\zeta)\in\R$ for any $\zeta\in\R$;
\item $(f(z))^2=O(|z-w_j|^{-1})$ as $z\to w_j$ and $\res{}_{z=w_j}(f(z))^2\in i\R_+$ for all $j$;
\item $f(z)=1+O(z^{-1})$ as $z\to\infty$.
\end{list}
Note that for $m=0$ (i.e., the trivial cover), we have an obvious solution $f_0\equiv 1$.

\smallskip

In order to find $f$, we introduce an auxiliary spinor
\[
f_{w_1,\dots,w_m}(z):=\cB_{w_1}(z)\cdot \dots\cdot \cB_{w_m}(z)\,,\qquad
\cB_{w}(z):=\frac{z-\Re \, w}{[(z\!-\!\overline{w})(z\!-\!w)]^{\frac{1}{2}}}\,.
\]
Note that it satisfies (\reflistAC), (\reflistBC) and (\reflistDC). Moreover, it has real zeros at $t_j=\Re\, w_j$; thus, (\reflistAC), (\reflistBC) and (\reflistDC) also hold for the product $f_{w_1,\dots,w_m}(z)\cdot g(z)$, where $g(z)$ is any function of the form
\begin{equation}
\label{g_as_fractions} g(z)\equiv 1+\sum_{j=1}^m \frac{\lambda_j}{t_j\!-\!z},\qquad
\lambda_k\in\R
\end{equation}
(and if some $t_j$ coinside, we can add higher-order poles, so that $g(z)$ is always a linear combination of $m$ linearly independent functions). We are looking for parameters $\lambda_j$ such that $f_{w_1,\dots,w_m}(z)\cdot g(z)$ satisfy (\reflistC) as well. Denote
\[
R_k:= [-2i\res\nolimits_{z=w_k}(f_{w_1,\dots,w_m}(z))^2]^{\frac{1}{2}}.
\]
Then, the condition (\reflistC) for $f_{w_1,\dots,w_m}(z)\cdot g(z)$ can be restated as
\begin{equation}
\label{system} \Im \left[R_k\cdot g(w_k)\right]=\Im R_k + \sum_{j=1}^m
\lambda_j\cdot\Im\frac{R_k}{t_j\!-\!w_k}=0\quad \text{for~all~}k=1,\dots,m\,.
\end{equation}

Note that (\ref{system}) is an $m\times m$ linear system in $\lambda_j$. We argue that this system is always non-degenerate.
Indeed, if $\lambda_j^0$ is a solution to the corresponding \emph{homogeneous} system, then $f(z):=[\sum_{j=1}^m \lambda_j^0/(t_j-z)]\cdot
f_{w_1,\dots,w_m}(z)$ is a spinor satisfying \mbox{(\reflistAC)--(\reflistCC)} and such that
$f(z)=O(z^{-1})$ as $z\to\infty$. But it follows from the proof of Lemma~\ref{lemma: uniqueness_bis} that any such spinor is identically zero (after mapping to a bounded domain, the condition at infinity yields boundedness near $a$). Thus, $\lambda_j^0\equiv 0$.

Taking into account that one has explicitly
\[
R_k=({\Im\,w_k})^{\!\frac{1}{2}}\cdot\prod_{j\ne k}\cB_{w_j}(w_k)=
({\Im\,w_k})^{\!\frac{1}{2}}\cdot\prod_{j\ne k}
\frac{(w_k\!-\!\Re\,w_j)}{[(w_k\!-\!\overline{w}_j)(w_k\!-\!w_j)]^{\frac{1}{2}}}\,.
\]
and solving the linear system (\ref{system}), one obtains $\lambda_j$ as ratios of certain explicit
$m\times m$ determinants (depending of $w_j$'s), and then the ratio of spin correlations (\ref{ThetaW1m}) is given by
\[
\vartheta(w_1,\dots,w_m)=f(0)= g(0)\cdot \prod_{j=1}^m\cB_{w_j}(0) =\biggl[1+\sum_{j=1}^m
\frac{\lambda_j}{\Re\,w_j}\biggr]\cdot \prod_{j=1}^m\frac{\Re\,w_j}{|w_j|}\,.
\]
Here we have used conventions of Remark \ref{rem: bcinversion} to determine the sheet of the double
cover of $\C_+\setminus\{w_1,\dots,w_n\}$ (that is, signs of the square roots):
in order to get the value $f(0)$, one starts with the value
$f(\infty)=+1$ and continuously moves the boundary point to $0$ along
the \emph{counterclockwise} boundary arc $(-\infty,0)$, thus arriving to
$\cB_{w_j}(0)=(-\Re\, w_j)/(-|w_j|)$ . The other way to fix the sign is given by
taking all $w_j$ close to the boundary arc $(0;+\infty)$: this should yield a positive correlation as that arc carries ``+'' boundary conditions.

%If some of $t_j=\Re\,w_j$ coincide, then one should either naturally modify (\ref{system})
%including higher order poles, or just take a limit of the answer obtained for $t_j\ne t_k$.

\begin{rem}
In particular, for a \emph{single} point $w\in\C$ (i.e., for the case $m=1$) one has
$t=\Re\,w$, $R=(\Im\,w)^{1/2}\in\R$, thus $\lambda=0$ and
\[
\vartheta(w)=%\cB_w(0)=\frac{\overline{w}+w}{2|w|}=
\frac{\Re\,w}{|w|}=\cos\,[\pi\text{hm}_{\C_+}(w,\R_-)]\,.
\]
Since harmonic measure is conformally invariant, this implies
\[
\frac{\E_{a^\mesh b^\mesh}[\,\sigma(w^\mesh)\,]} {\E_+[\,\sigma(w^\mesh)\,]} \to
\cos\,[\pi\text{hm}_{\Omega}(w,(ab))]\,,
\]
if $(\Od;a^\mesh,b^\mesh)$ approximate $(\Omega;a,b)$ and faces $w^\mesh$ tend to an inner
point $w\in\Omega$.
\end{rem}

%\bibliography{sle}
%\bibliographystyle{alpha}
%\end{document}

\def\cprime{$'$}

\end{document}